\documentclass[british,a4paper,10pt]{article}
\usepackage{babel,amsmath,amssymb,amsthm}
\usepackage{bm}
\usepackage[sans]{dsfont}
\usepackage[mathscr]{euscript}
\usepackage{times}
\usepackage{cite}
\usepackage{xcolor}
\usepackage{mathtools}
\usepackage[all]{xy}
\usepackage{float}
\usepackage{hyperref}
\usepackage{a4wide}
\theoremstyle{plain}
\newtheorem{theorem}{Theorem}
\newtheorem{proposition}[theorem]{Proposition}

\theoremstyle{remark}
\newtheorem{remark}{Remark}

\theoremstyle{definition}
\newtheorem{definition}{Definition}
\theoremstyle{definition}
\newtheorem{assumption}{Assumption}

\newcommand{\Tr}{\operatorname{Tr}}
\newcommand{\rmd}{\mathrm{d}}
\newcommand{\rme}{\mathrm{e}}
\newcommand{\rmi}{\mathrm{i}}
\newcommand{\RE}{\mathrm{\,Re\,}}
\newcommand{\IM}{\mathrm{\,Im\,}}

\newcommand{\Rbb}{\mathbb{R}}
\newcommand{\Cbb}{\mathbb{C}}
\newcommand{\Sbb}{\mathbb{S}}
\newcommand{\Pbb}{\mathbb{P}}

\newcommand{\Ibb}{\mathbb{I}}
\newcommand{\T}{\mathtt{T}}
\newcommand{\Qbb}{\mathbb{Q}}
\newcommand{\Ebb}{\operatorname{\mathbb{E}}}

\newcommand{\openone}{\mathds{1}}

\newcommand{\ind}{\mathtt{1}}
\newcommand{\id}{\mathrm{Id}}

\newcommand{\norm}[1]{\left\Vert#1\right\Vert}
\newcommand{\abs}[1]{\left\vert#1\right\vert}

\newcommand{\Acal}{\mathcal{A}}
\newcommand{\Bcal}{\mathcal{B}}

\newcommand{\Gcal}{\mathcal{G}}

\newcommand{\Ical}{\mathcal{I}}

\newcommand{\Kcal}{\mathcal{K}}
\newcommand{\Lcal}{\mathcal{L}}

\newcommand{\Tcal}{\mathcal{T}}

 \newcommand{\Ascr}{\mathscr{A}}
\newcommand{\Bscr}{\mathscr{B}}

 \newcommand{\Escr}{\mathscr{E}}
\newcommand{\Fscr}{\mathscr{F}}

\newcommand{\Hscr}{\mathscr{H}}

\newcommand{\Sscr}{\mathscr{S}}
\newcommand{\Tscr}{\mathscr{T}}
\newcommand{\Uscr}{\mathscr{U}}

\newcommand{\Xscr}{\mathscr{X}}

 \newcommand{\Eo}{\mathsf{E}}
\newcommand{\Fo}{\mathsf{F}}

\newcommand{\Uo}{\mathsf{U}}

\newcommand{\marginlabel}[1]{\mbox{}\marginpar{#1}}

\begin{document}
\title{Markovian dynamics for a quantum/classical system \\ and quantum trajectories}

\author{Alberto Barchielli \footnote{also
Istituto Nazionale di Alta Matematica (INDAM-GNAMPA)} \\ Istituto Nazionale di Fisica Nucleare (INFN), Sezione di Milano, Italy }

\maketitle

\begin{abstract}
Quantum trajectory techniques have been used in the theory of open systems as a starting point for numerical computations and to describe the monitoring of a quantum system in continuous time. We extend this technique to develop a general approach to the dynamics of quantum/classical hybrid systems. By using two coupled stochastic differential equations, we can describe a classical component and a quantum one which have their own intrinsic dynamics and which interact with each other. A mathematically rigorous construction is given, under the restriction of having a Markovian joint dynamics and of involving only bounded operators on the Hilbert space of the quantum component. An important feature is that, if the interaction allows for a flow of information from the quantum component to the classical one, necessarily  the dynamics is dissipative. We show also how this theory is connected to a suitable hybrid dynamical semigroup, which reduces to a quantum dynamical semigroup in the purely quantum case and includes Liouville and Kolmogorov-Fokker-Planck equations in the purely classical case. Moreover, this semigroup allows to compare the proposed stochastic dynamics with various other proposals based on hybrid master equations. Some simple examples are constructed in order to show the variety of physical behaviours which can be described; in particular, a model presenting hidden entanglement is introduced.
\end{abstract}

\vspace{2pc}

\tableofcontents

\section{Introduction}

The interest in quantum/classical hybrid systems is old, see for instance \cite{DGS00,DamWer23,Bar23,BarW23,Op+22,Dio23,ManRT23,Opp+23,Sergi+23,Pomar+23,Tronci23,BriU23,Dio14,LOW22} and references therein. One of the main motivations of the study of hybrid systems is that the output of a monitored system is classical; then, implicitly or explicitly, the dynamics of quantum/classical systems is involved in the theory of quantum measurements in continuous time and quantum filtering \cite{Bar86,Bar93,BarG09,Bar06,DGS00,Dio14,BarB91,Hol01,ZolG97,WisM10,Bel88,Bel89,BarPZ98,Op+22,Dio23,Till24,BGM04,Maa23}. Moreover, hybrid systems could be used as an approximation to complicate quantum systems, as an effective theory \cite{DGS00,Op+22,ManRT23,Tronci23,BriU23,Pomar+23,Mim24}. Another motivation is in the connections between a theory of gravity and quantum theories; perhaps gravity is ``classical'', and a dynamical theory of gravity and matter needs a hybrid dynamics \cite{HallR18,OppWD22,Opp+23grav,Op+22,DGS00,ProsB23,Till24,LOW22}. Classical environments could be involved even in revivals of quantum entanglement \cite{LoFC18,BarGent}. A common feature in many of these attempts is that a non-trivial hybrid dynamics is always irreversible, due to the flow of information from the quantum system  to the classical one \cite{DGS00,Opp+23,Op+22,Opp+23grav,Bar23,BarW23,DamWer23,LOW22,Dio23}.

The main aim of this article is to construct a consistent formulation of the dynamics of quantum/classical hybrid systems, by connecting the quantum trajectory theory, used to describe open and monitored quantum systems
\cite{BarH95,BarG12,BarG09,WisM10,Bel88,Bel89,BarPZ98}, to the classical theory of Markov processes and stochastic differential equations \cite{Met82,IWat89,LipSmart,DaPZ14,Prott04,App09,EthK05}. Some proposals of connecting quantum trajectories to hybrid dynamics were given in \cite{LOW22,Till24,Dio23}. The dynamics of the hybrid systems is given in Sec.\ \ref{sec:SDE} in terms of coupled stochastic differential equations; only the Markovian case is considered. We also show how this approach is connected to notions appearing in the general formulation of quantum measurements, such as ``positive operator valued measures'' and ``instruments'' \cite{WisM10,Hol01,BarG09}.

In Sec.\ \ref{sec:hds} we show how to associate a suitable semigroup to the constructed dynamics. This hybrid semigroup reduces to a quantum dynamical semigroup \cite{Hol01,BarG09,WisM10} in the case of a pure quantum system, while in the pure classical case it includes the Liouville and the Kolmogorov-Fokker-Planck equations \cite{App09}. We give also the formal generator of this semigroup, which allows to compare the class of dynamics developed here with the ones obtained by other approaches. This generator also shows analogies and links with the theory of infinitely divisible distributions \cite{App09,Sato99,Hol88QPIII,BarW23}.

To illustrate the possibilities of the constructed dynamics, various examples are presented in Sec.\  \ref{sec:ex}. In particular a two qubit system is considered and analogies with the phenomenon of hidden entanglement \cite{LoFBen14,LoFC18} are discussed. Some final conclusions are given in Sec.\ \ref{sec:concl}.

We end this section by some preliminary notions and useful notations.

\subsection{Some notations}\label{sec:notations}

The quantum component \marginlabel{$\Hscr$, $\Bscr(\Hscr)$} will be represented in a separable complex Hilbert space $\Hscr$. The bounded linear operators are denoted  by  $\Bscr(\Hscr)$  \marginlabel{$\Tscr(\Hscr)$, $\openone$}  and the trace class by $\Tscr(\Hscr)$; the unit element in $\Bscr(\Hscr)$ is denoted by $\openone$ \marginlabel{$a,a^\dagger$, $\alpha,\overline{\alpha}$} and the adjoint of the operator $a\in\Bscr(\Hscr)$ by $a^\dagger$.
The complex conjugated of $\alpha\in \Cbb$ is denoted by $\overline \alpha$.  We shall use also the notation \marginlabel{$\langle \rho, a\rangle$}
\[
\langle \rho, a\rangle=\Tr\{ \rho a\}, \qquad \rho\in \Tscr(\Hscr), \quad a\in \Bscr(\Hscr).
\]
By $\Sscr(\Hscr)\subset \Tscr(\Hscr)$ \marginlabel{$\Sscr(\Hscr)$} we denote the set of the statistical operators (the \emph{quantum states}). Then, general measurements are represented by \emph{positive operator valued measures} and \emph{instruments} \cite{BarG09,WisM10,Hol01}.

\begin{definition} \label{def:instr} Let $\Xscr$ be a set and $\Sigma_\Xscr$  a $\sigma$-algebra of subsets of $\Xscr$. An \emph{instrument} $\Ical(\cdot)$ with value space $(\Xscr,\Sigma_\Xscr)$ is a function from the $\sigma$-algebra $\Sigma_\Xscr$  to the linear maps on $\Bscr(\Hscr)$ such that: \quad (i) (positivity) for every set $\Eo\in \Sigma_\Xscr$, $\Ical(\Eo)$ is a completely positive and normal map from $\Bscr(\Hscr)$ into itself;
\quad (ii) ($\sigma$-additivity)  we have $\Ical\left(\bigcup_i \Eo_i\right)[a]=\sum_i\Ical(\Eo_i)[a]$, $\forall a\in \Bscr(\Hscr)$ and for every countable family of disjoint sets $\Eo_i\in \Sigma_\Xscr$; \quad (iii) (normalization) $\Ical(\Xscr)[\openone ]=\openone$.
\end{definition}
``Normality'' is a suitable regularity requirement; for a positive map it is equivalent to require the map to be the adjoint of a bounded map on $\Tscr(\Hscr)$.
Let us recall that an instrument gives the probabilities for the observed values and the conditional state after the measurement. The quantity $\Ical(\cdot)[\openone]$ is a positive operator valued measure, and it gives only the probabilities. Moreover, measurements on a quantum system can be interpreted as involving hybrid systems: a positive operator valued measure is a channel from a quantum system to a classical one, an instrument is a channel from a quantum system to a hybrid system \cite{BarL08,Dio14,DamWer23,ManRT23}.
When \marginlabel{$\Rbb^s$} $\Xscr=\Rbb^s$  it is usual to take for $\Sigma_\Xscr$ the $\sigma$-algebra of the Borel sets, denoted by $\Bcal(\Rbb^s)$. \marginlabel{$\Bcal(\Rbb^s)$}

We take a classical component with values in $\Rbb^s$, the classical \emph{phase space} \cite{DamWer23,BarW23}. The choice of observable and state spaces in the classical case is more involved \cite[Sec.\ 1.1]{DamWer23}. To introduce the semigroups associated to the Markov processes on which our classical component will be based, it is convenient to use  the space $B_{\rm b}(\Rbb^s)$ \marginlabel{$B_{\rm b}(\Rbb^s)$} of bounded, Borel measurable complex functions
on $\Rbb^s$ \cite{EthK05,Prott04,App09} and its subspace $C_{\rm b}(\Rbb^s)$ \marginlabel{$C_{\rm b}(\Rbb^s)$}  of the continuous bounded functions; both spaces are  Banach spaces under the supremum  norm \cite[p.\ 6]{App09}.
We consider complex functions instead of only real ones, because it is more convenient in constructing hybrid systems.  Inside a $C^*$-algebraic approach, $C_{\rm b}(\Rbb^s)$ was one of the suggestions for hybrid systems given in \cite[p. 7]{DamWer23}.

\section{Hybrid dynamics and stochastic differential equations}\label{sec:SDE}
In the quantum theory of measurements in continuous time, many conceptual and technical results were obtained by the quantum trajectory approach, based on stochastic differential equations (SDEs) \cite{Hol01,BarH95,BarPZ98,BarG09,BarG12,WisM10}.
By extending this technique, it is possible to develop a quantum trajectory formulation of a very general class of hybrid dynamics.

The construction is based on two probability measures: a \emph{reference probability} $\Qbb$ and a \emph{physical probability} $\Pbb$. The physical probability will be constructed in Sec.\ \ref{sec:nprob}, in terms of $\Qbb$ and of the quantum/classical dynamics.

In the next Assumption we introduce the reference probability and the Wiener and Poisson processes needed in the construction. They are defined in a filtered probability space satisfying the \emph{usual hypothesis}, typically assumed in books on stochastic calculus  \cite[p.\ 3 and Sec.\ I.5]{Prott04}.

\begin{assumption}\label{Ass:processes}
Let $\big(\Omega, \Fscr, \Qbb\big)$ \marginlabel{$\left(\Omega, \Fscr, \Qbb\right)$} be a \emph{complete probability space} with a \emph{filtration} of $\sigma$-algebras $\left\{\Fscr_t, t\geq 0\right\}$, \marginlabel{$\Fscr_t,\;\Fscr_{t_-}$} such that \ (a) \ $\Fscr= \bigvee_{t\geq 0} \Fscr_t$, \ (b) \ $\Fscr_t= \bigcap_{r> t}\Fscr_r$, \ (c) \ $\Fo\in\Fscr, \ \Qbb(\Fo)=0 \ \Rightarrow \ \Fo\in \Fscr_0$; we set also $\Fscr_{t_-}=  \bigvee_{0\leq r<t} \Fscr_r$ \cite[pp.\ 3-4]{Met82}, \cite[p.\ 45]{IWat89}.

Let $W_k$, $k=1,\ldots,d$, \marginlabel{$W_k(t)$} be continuous versions of independent, adapted, standard \emph{Wiener processes},  with increments independent of the past \cite[Sec.\ I.7]{IWat89}.

Let $\Uscr$ \marginlabel{$(\Uscr,\Bcal(\Uscr))$} be a locally compact Hausdorff space with a topology with a countable basis and $\Bcal(\Uscr)$   be the Borel $\sigma$-algebra of $\Uscr$. Let $\nu$ \marginlabel{$\nu(\rmd u)$, $\mathsf{U}_0$} be a Radon measure on $(\Uscr,\Bcal(\Uscr))$ with  the property:
\begin{equation}
\exists \mathsf{U}_0 \in \Bcal(\Uscr): \ \nu(\Uscr\setminus \mathsf{U}_0)<+\infty.
\end{equation}
Finally, $\Pi(\rmd u,\rmd t)$ \marginlabel{$\Pi(\rmd u,\rmd t)$} is
an adapted  \emph{Poisson point process} on $\Uscr\times \Rbb_+$ of intensity $\nu(\rmd u)\rmd t$; $\Pi$ is independent of the Wiener processes and with increments independent of the past \cite[Sec.\ 1]{BarPZ98}, \cite[Secs.\ I.8, I.9]{IWat89}.
\end{assumption}

It is also useful to introduce the \marginlabel{$\widetilde \Pi(\rmd u,\rmd t)$} \emph{compensated Poisson measure}:
\[
\widetilde \Pi(\rmd u,\rmd t):=\Pi(\rmd u,\rmd t)-\nu(\rmd u)\rmd t.
\]
Stochastic integrals with respect to the Wiener and Poisson processes are defined, for instance, in \cite{Met82,IWat89,LipSmart,DaPZ14,Prott04}.
In particular, the stochastic integrals with respect to $\Pi(\rmd u,\rmd t)$ and with respect to $\widetilde \Pi(\rmd u,\rmd t)$ are defined in \cite[Sec.\ II.3]{IWat89}, where also the classes of possible integrands are introduced.
In the following we shall make use of It\^o's formula and generalizations; see
 \cite[Sec.\ II.5]{IWat89}, \cite[Sec.\ I.4]{LipSmart}, \cite[Sec.\ II.7]{Prott04}. The simpler case of a finite collection of Poisson processes is presented in Sec.\ \ref{sec:Udiscr}.

We shall denote by $\Ebb_\Qbb[X]$ \marginlabel{$\Ebb_\Qbb$}the mean value of the random variable $X$ with respect to the probability measure $\Qbb$. As a real random variable is a measurable function from $\Omega$ to $\Rbb$ we can write $X: \omega \in \Omega \to X(\omega)\in \Rbb$; then, the mean value can be written as $\Ebb_\Qbb[X]=\int_\Omega X(\omega) \Qbb(\rmd \omega)$.

\subsection{The classical component}
The classical component of the hybrid system is the stochastic process $X(t)$ taking values in the phase space $\Rbb^s$ and satisfying the  stochastic differential equation
\begin{multline}
\rmd X_i(t) =c_i \big(X(t_-)\big)\rmd t +\sum_{k=1}^{d}b_{i}^k\big(X(t_-)\big) \rmd W_k(t)+\int_{\mathsf{U}_0} g_{i}\big(X(t_-),u\big) \widetilde\Pi(\rmd u,\rmd t)
\\ \label{Xprocess} {}
+ \int_{\Uscr\setminus \mathsf{U}_0} g_{i}\big(X(t_-),u\big) \Pi(\rmd u,\rmd t), \qquad i=1,\ldots, s.
\end{multline}
When considering stochastic processes with diffusion and jumps, it is useful to fix some regularity properties for their trajectories. We always understand to have \emph{regular right continuous}  processes, which means processes adapted to the given filtration $\{\Fscr_t\}_{t\geq 0}$,  with trajectories continuous from the right and with limits from the left \cite[Def.\ 1.5]{Met82}; indeed, $X(t_-)$ means the limit from the left.

\begin{assumption}[Growth and Lipschitz conditions]\label{ass:classical}
The $b^k$ and  $c$ are Borel measurable functions from $\Rbb^s$ to $\Rbb^s$; the $g$ are measurable functions from $\Rbb^s\times \Uscr$ into $\Rbb^s$. We have also, for some $K>0$, and $\forall x, y \in \Rbb^s$,
\begin{equation*}
\norm{c(x)}^2+\sum_k\norm{b^k(x)}^2 +\int_{\mathsf{U}_0}\norm{g(x,u)}^2 \nu(\rmd u)\leq K\left(1+\norm{x}^2\right),
\end{equation*}
\begin{equation*}
\norm{c(x)-c(y)}^2+\sum_k\norm{b^k(x)- b^k(y)}^2 +\int_{\mathsf{U}_0}\norm{g(x,u)-g(y,u)}^2 \nu(\rmd u)\leq K\norm{x-y}^2.
\end{equation*}
\end{assumption}
\begin{proposition}\label{prop:clpr}
Under Assumptions \ref{Ass:processes} and \ref{ass:classical}, given a random ($\Fscr_0$-measurable) initial condition $X(0)$, there is a unique regular right continuous process solving \eqref{Xprocess}; moreover, it turns out to be a time-homogeneous Markov process. In particular, $\forall x\in \Rbb^s$ there is a unique solution $X^x(t)$ \marginlabel{$X^x(t)$} with deterministic initial condition $X^x(0)=x$.
\end{proposition}
The proof of this Proposition is given in \cite[Sec.\ IV.9]{IWat89}; see, in particular, Theorem 9.1 of \cite{IWat89}. Note that no condition is given on $g(x,u)$ for $u\in \Uscr\setminus \Uo_0$, apart from being jointly measurable in $(x,u)$. This is due to the fact that both $\Pi$ and $\nu$ are finite measures on $\Uscr\setminus \Uo_0$.

It is usual to call \emph{transition probability} (or transition function) \cite{EthK05,App09,Prott04} the function on $[0,+\infty) \times \Rbb^s \times \Bcal(\Rbb^s )$, defined by
\begin{equation}\label{transprob}
P(t,x,\Eo)=\Qbb[X^x(t)\in\Eo], \qquad t\in[0,+\infty), \quad x\in\Rbb^s, \quad \Eo\in \Bcal(\Rbb^s ).
\end{equation}
Some properties of the solutions of  \eqref{Xprocess} and of the transition probabilities are collected in \ref{app:classM}, where also the associated semigoup is given.

If for some $i$ we have $c_i=0$, $g_i=0$, and $b_{i}^j=\delta_{ij}$,  the corresponding Wiener process $W_i$ is among the observed classical processes. If for some $i$ we have $c_i=0$, $b_{i}^\cdot=0$, and $g_{i}(x,u)=\ind_\mathsf{U}(u)$, with $\nu(\mathsf{U})<+\infty $, the Poisson process  $\Pi(\mathsf{U},(0,t])$ is among the observed classical processes.

\subsection{Stochastic Schr\"odinger equations}\label{sec:SSE}

We introduce now the stochastic equation which determines the dynamics of the quantum component and the interaction of the two components. The coefficients involved in such an equation are given in the following assumption; to avoid analytical complications, we consider only bounded operators on the Hilbert space of the quantum component.
\begin{assumption}\label{Ass:1}
Let $H(x)$, $L_k(x)$, $J(x,u)$, $k=1,\ldots,d$, \ $x\in \Rbb^s$, $u\in\Uscr$, \marginlabel{$H(x)$} \marginlabel{$L_k(x)$} \marginlabel{$J(x,u)$} be bounded operators on the Hilbert space $\Hscr$ of the quantum component. All these functions are continuous in $x$ in the strong operator topology; $(x,u) \mapsto J(x,u)$ is strongly measurable. We also have
\begin{equation}\label{somecond}\begin{split}
H(x)^\dagger = H(x), \qquad  &\sup_{x\in \Rbb^s} \norm{H(x)}<+\infty, \qquad \sup_{x\in \Rbb^s} \norm{L_k(x)}<+\infty, \qquad \sup_{x\in \Rbb^s,\; u\in \Uscr} \norm{J(x,u)}<+\infty,
\\
&D(x):=\int_\Uscr J(x,u)^\dagger J(x,u)\,\nu(\rmd u) \in \Bscr(\Hscr),  \qquad \sup_{x\in \Rbb^s} \norm{D(x)}<+\infty;
\end{split}
\end{equation}
the integral is defined in the weak topology of $\Bscr(\Hscr)$. \marginlabel{$D(x)$}
\end{assumption}

These operators become operator-valued random variables when composed with $X(t_-)$. Indeed, we define, for $x\in \Rbb^s$ and $\omega \in \Omega$, \marginlabel{$K(x)$}
\begin{equation}\label{randomcoeff}
K(x)=\rmi H(x)  +\frac{1}2\sum_{k=1}^{d} L_k(x)^\dagger  L_k(x)+\frac 12 \, D(x),
\end{equation}
\begin{equation}\label{LtJt}\begin{split}K_t(\omega)=K\big(X(t_-;\omega)\big), \qquad  &H_t(\omega)=H\big(X(t_-;\omega)\big),
\\
L_{kt}(\omega)=L_k\big(X(t_-;\omega)\big),\qquad J_t(u;\omega)=& J\big(X(t_-;\omega),u\big), \qquad D_t(\omega)= D\big(X(t_-;\omega)\big).
\end{split}
\end{equation}
We have inserted $X(t_-)$ in order to have \emph{predictable} operator-valued processes, as was for the coefficients in Eq.\ \eqref{Xprocess}. Let us stress that these coefficients depend on the solution of \eqref{Xprocess} and, so, on the initial condition $X(0)$.

Finally, we introduce the \emph{linear stochastic Schr\"odinger equation} (SSE) in $\Hscr$ \cite[(2.3)]{BarH95} \marginlabel{SSE}
\begin{equation}\label{lSSE}
\rmd  \psi_t= -K_t \psi_{t_-}\rmd t + \sum_{k=1}^d L_{kt}\psi_{t_-}\, \rmd  W_k(t)+ \int_\Uscr J_t(u)\psi_{t_-}\,\widetilde \Pi(\rmd u,\rmd t),  \quad \Ebb_{\Qbb}\big[\norm{\psi_0}^2\big]<+\infty;
\end{equation}
the random initial condition $\psi_0\in\Hscr$ is $\Fscr_0$-measurable.

Let us note that the assumption of ``no memory'' in the dynamics of the composed system is implicitly contained in the fact that the coefficients in \eqref{Xprocess} and \eqref{lSSE} depend on the past only through $X(t_-)$ and that the involved noises are Wiener and Poisson processes, which have increments independent of the past. Examples without these hypotheses are considered in \cite{BarH95}.
However, with respect to what is done in that reference, the space $\Uscr$ involved in the point process is slightly more general; similarly, the structure of the classical process $X(t)$ has a jump contribution more general than in \cite[(4.32)]{BarH95}. In any case, the proofs of that article apply to this case too and we can follow that construction.

\begin{theorem}
Up to $\Qbb$-equivalence, Eq.\ \eqref{lSSE} admits a unique  solution $\psi_t$, $t\in \Rbb^+$, which is an $\Hscr$-valued semimartingale \cite[Sec.\ II.1]{Prott04}. The process $\norm{\psi_t}^2$ is a non-negative martingale \cite[Sec.\ I.2]{Prott04}, such that
\marginlabel{$\psi_t, \ \widehat\psi_t$}
\[
\Ebb_\Qbb\big[\norm{\psi_t}^2\big]=\Ebb_\Qbb\big[\norm{\psi_0}^2\big]<+\infty,
\qquad \rmd \norm{\psi_t}^2=\norm{\psi_{t_-}}^2\,\rmd Z_t,
\]
\[
Z_t=\sum_{k=1}^d\int_{(0,t]} m_{ks} \,\rmd W_k(s) +\int_{(0,t]}\int_\Uscr \left( I_s(u) -1\right) \widetilde\Pi(\rmd u, \rmd s),
\quad m_{kt}= 2 \RE \langle \widehat \psi_{t_-} \big| L_{kt} \widehat \psi_{t_-} \rangle,
\]
\[
I_t(u)=\norm{\big(J_t(u)+\openone\big)\widehat \psi_{t_-}}^2,
\qquad \widehat \psi_{t}(\omega)=\begin{cases}\psi_t(\omega)/\norm{\psi_t(\omega)} & \text{if } \  \norm{\psi_t(\omega)}\neq 0,
\\ v \ {\rm (fixed\ unit\ vector)} & \text{if } \ \norm{\psi_t(\omega)}= 0.\end{cases}
\]
\end{theorem}
\begin{proof}
Assumptions \ref{Ass:processes}, \ref{Ass:1}, and the properties of $X(t)$ easily imply Properties 2.0.A, 2.0.B, 2.2, 2.3.A, 2.3.B in \cite{BarH95}; then, Proposition 2.1 and Theorem 2.4 of \cite{BarH95} hold, which is the thesis of the theorem above.
\end{proof}

\begin{remark}[The non-linear SSE]\label{rem:nlSSE} The process $\widehat \psi_t$, $0\leq t\leq T$, satisfies a non-linear stochastic differential equation, the \emph{non-linear SSE}, under the new probability $\norm{\psi_T(\omega)}^2\Qbb(\rmd \omega)$, see \cite[Sec.\ 2.2]{BarH95}.
A first example of SSE goes back to \cite{Bel88}.\end{remark}

\subsection{Stochastic master equations}\label{sec:sigmat}
We introduce now the stochastic differential equation in $\Tscr(\Hscr)$ associated to the linear SSE \eqref{lSSE}.

\paragraph{Initial condition $\sigma_0$.} Let $\psi^\beta\in \Hscr$ be $\Fscr_0$-measurable random vectors  such that $\Ebb_\Qbb\left[\sum_{\beta=1}^{+\infty} \norm{\psi^\beta}^2\right]=1$. Then, we define the random trace-class operator $\sigma_0$ by \marginlabel{$\sigma_0$}
\begin{equation}\label{inst}
\langle \sigma_0,a\rangle=\sum_{\beta=1}^{+\infty}\langle \psi^\beta|a \psi^\beta\rangle , \quad \forall a\in \Bscr(\Hscr) \qquad \Rightarrow \qquad \sigma_0\geq 0, \quad \Ebb_\Qbb\big[\Tr\{\sigma_0\}\big]=1.
\end{equation}

Let $\psi_t^\beta$ be the solution of \eqref{lSSE} with initial condition $\psi^\beta$ and define $\sigma_t$ by \marginlabel{$\sigma_t$}
\begin{equation}\label{def:sigmat}
\langle \sigma_t,a\rangle =\sum_{\beta=1}^{+\infty}\langle \psi_t^\beta|a \psi_t^\beta\rangle, \qquad \forall a\in \Bscr(\Hscr).
\end{equation}
By the formal rules of stochastic calculus, we get the \emph{linear stochastic master equation} \marginlabel{SME} (SME) for $\sigma_t$: $\forall a\in \Bscr(\Hscr)$, \marginlabel{$\Lcal_t$} \marginlabel{$\Lcal(x)$}
\begin{multline}
\rmd \langle \sigma_t,a \rangle= \langle \sigma_{t_-},\Lcal_t[a] \rangle \rmd t +\sum_{k=1}^d \langle \sigma_{t_-},{L_{kt}}^\dagger a +aL_{kt}\rangle \rmd W_k(t)
\\ \label{linearSME}{}+ \int_{\Uscr}\langle \sigma_{t_-}, J_t(u)^\dagger a J_t(u)+ J_t(u)^\dagger a+a J_t(u)\rangle \widetilde \Pi(\rmd u,\rmd t),
\end{multline}
\begin{equation}\label{Lcalt}
\Lcal_t=\Lcal\big(X(t_-)\big),
\end{equation}
\begin{multline}
\Lcal(x)[a]=
\rmi[H(x),a]+\sum_{k=1}^d\left( L_{k}(x)^\dagger a L_{k}(x)-\frac 12 \left\{L_{k}(x)^\dagger  L_{k}(x),a\right\}\right)
\\ \label{Lcalx}{}+\int_\Uscr \left(J(x,u)^\dagger a J(x,u)-\frac 12 \left\{J(x,u)^\dagger  J(x,u),a\right\}\right) \nu(\rmd u).
\end{multline}
Let us stress that the SME \eqref{linearSME} is \emph{linear} once the process $X(t)$, solving \eqref{Xprocess}, as been fixed. Instead, if we consider \eqref{Xprocess} and \eqref{linearSME} as a system of SDEs, we have the non-linearity introduced by \eqref{Xprocess}.

\begin{theorem}\label{theor:sigmaprop}
The $\Tscr(\Hscr)$-valued process $\sigma_t$, defined in \eqref{def:sigmat}, satisfies the weak-sense stochastic differential equation \eqref{linearSME}. Moreover, it is the unique solution of \eqref{linearSME} with the initial condition \eqref{inst}.
We have also \marginlabel{$p_t$} \marginlabel{$m_k(t)$} \marginlabel{$I_t(u)$} \marginlabel{$\widehat \sigma_t$}
\begin{eqnarray}\label{def:pt}
\sigma_t\geq 0, \qquad p_t:=\langle \sigma_t,\openone\rangle \equiv \Tr \{\sigma_t\}\geq 0, \qquad \Ebb_\Qbb[p_t]=1,
\\ \label{eq:pt}
\rmd p_t=p_{t_-}\biggl(\sum_{k=1}^d m_{k}(t) \,\rmd W_k(t) +\int_\Uscr \left(I_t(u) -1\right) \widetilde\Pi(\rmd u, \rmd t)\biggr),
\\ \label{defmk}
m_{k}(t)= 2 \RE \langle \widehat \sigma_{t_-} , L_{kt} \rangle, \qquad I_t(u)=\big\langle \widehat \sigma_{t_-}, \big(J_t(u)^\dagger +\openone\big)\big(J_t(u)+\openone\big)\big\rangle,
\\ \label{defhatsigma}
\widehat \sigma_{t}(\omega)=\begin{cases}\sigma_t(\omega)/p_t(\omega) & \text{if } \  p_t(\omega)\neq 0,
\\ \sigma_v \ {\rm (fixed\ statistical\ operator)} & \text{if } \ p_t(\omega)= 0.\end{cases}
\end{eqnarray}
\end{theorem}

\begin{proof}
Again Assumptions \ref{Ass:processes}, \ref{Ass:1}, and the properties of $X(t)$ imply the validity of
Properties 3.1 and 3.3 of \cite{BarH95}. Then, Propositions 3.2 and 3.4 of \cite{BarH95} give the thesis of the theorem.
\end{proof}
By definition $\widehat\sigma_t(\omega)\in\Sscr(\Hscr)$, i.e.\ $\widehat\sigma_t$ is a random quantum state. The random quantities $m_k(t)$ and $I_t(u)$ are defined by \eqref{LtJt}, \eqref{defmk},  so as to be $\Fscr_{t_-}$-measurable (\emph{predictable processes} \cite[Sec.\ 3]{Met82}). By construction, if the initial state $\widehat \sigma_{0}$ is a pure state, then $\widehat \sigma_{t}(\omega)$ too is a pure state.

\subsubsection{The physical probability and the classical component.}\label{sec:nprob}

As in \cite[Sec.\ 3.1]{BarH95}, we use $p_t$ \eqref{def:pt} as a probability density and we define the new probability $\Pbb$ on the $\sigma$-algebra $\Fscr=\bigvee_{t\geq 0} \Fscr_t$ by
\begin{equation}\label{physprob}
\forall t\geq 0, \quad \forall \Fo\in \Fscr_t, \qquad \Pbb (\Fo) := \Ebb_\Qbb\big[p_t\ind_\Fo\big];
\end{equation}
for a similar construction of a probability measure see \cite[Sec.\ 30.2]{Met82}.

Note that $p_t$ (the trace of $\sigma_t$) is a probability density with respect to $\Qbb$ and that it gives the probability $\Pbb$ on $\Fscr_t$, which contains the events related to the time interval $(0,t]$. By \eqref{eq:pt}, $p_t$ is a \emph{martingale} and this implies the consistency of the probabilities constructed in different time intervals: for $0<s<t$, $\Eo\in \Fscr_s$, one has $\int_\Eo p_t(\omega) \Qbb(\rmd\omega)=\int_\Eo p_s(\omega) \Qbb(\rmd\omega)$. This allows to extend $\Pbb$ to the whole $\Fscr$.

\begin{remark} The probability $\Pbb$, \marginlabel{$\Pbb$} defined by \eqref{physprob}, represents the physical probability of the hybrid system. By  \eqref{LtJt} and \eqref{linearSME}, the physical probability depends on the dynamics of both the classical component and the quantum one. Moreover, the probability $\Pbb$ depends on the initial conditions $X(0)$ and $\sigma_0$ used in solving  equations \eqref{Xprocess} and \eqref{linearSME}.
\end{remark}

\begin{remark} By the so called Girsanov transformation and generalizations \cite[Sec.\ III.8]{Prott04}, we have that, under the probability $\Pbb$, the processes
\begin{equation}\label{newW}
\widehat W_k(t):= W_k(t)- \int_{(0,t]} m_k(s) \rmd s
\end{equation}
are independent standard Wiener processes adapted to the filtration $\{\Fscr_t,\;t\geq 0\}$; moreover,
$\Pi(\rmd u,\rmd t)$ is no more a Poisson process, but it becomes a point process with compensator $I_t(u)\nu(\rmd u)\rmd t$  \cite[Prop.\ 2.5, Remarks 2.6, 3.5]{BarH95}. We also introduce the notation
\begin{equation}\label{hatPi}
\widehat\Pi(\rmd u,\rmd t)=\Pi(\rmd u,\rmd t)- I_t(u)\nu(\rmd u)\rmd t.
\end{equation}
According to \cite[Sec.\ 31]{Met82}, $\widehat\Pi(\rmd u,\rmd t)$ is a \emph{white random measure}.
\end{remark}

By rewriting \eqref{Xprocess} in terms of the new processes \eqref{newW} and \eqref{hatPi}, we obtain the dynamics of the classical component  under the physical probability $\Pbb$:
\begin{multline}
\rmd X_i(t) =\left(c_i \big(X(t_-)\big)+C^{(1)}_i(t)+ C^{(2)}_i(t) \right)\rmd t
+\sum_{k=1}^{d}b_{i}^k\big(X(t_-)\big) \rmd \widehat W_k(t)
\\ \label{XnewPa}{}+\int_{\mathsf{U}_0} g_{i}\big(X(t_-),u\big) \widehat\Pi(\rmd u,\rmd t)
+ \int_{\Uscr\setminus \mathsf{U}_0} g_{i}\big(X(t_-),u\big) \Pi(\rmd u,\rmd t),
\end{multline}
\begin{equation}\label{XnewP}
C^{(1)}_i(t)=\sum_{k=1}^{d}b_{i}^k\big(X(t_-)\big)m_k(t),\qquad
C^{(2)}_i(t)=\int_{\Uo_0} g_{i}\big(X(t_-),u\big)\left( I_t(u)-1\right)\nu(\rmd u).
\end{equation}
Here, $\Pi(\rmd u,\rmd t)$ is no more a Poisson measure and the processes $m_k(t)$ and $I_t(u)$ depend on the whole past through the quantum state; in particular, the dynamics \eqref{XnewPa} of the classical component depends on the behaviour of the quantum one.

From Eqs.\ \eqref{XnewPa}, \eqref{XnewP} we have that the classical component can extract some information from the quantum one only when not all the terms $C^{(1)}(t)$, $C^{(2)}(t)$, $I_t(u)$ vanish and they effectively depend on the quantum state. So, firstly we need  $b_{i}^k(x)$ and $g_{i}(x,u)$ not all identically vanishing and this means that the diffusive and the jump terms are not both vanishing: some dissipation must be present in the dynamics \eqref{XnewPa} of the classical component. Moreover, to have dependence on the quantum state at least some of the operators $L_k(x)$ and $J(x,u)$ must be not reducible to something proportional to the identity operator; this means that the Liouville operator \eqref{Lcalx} must contain a dissipative contribution. As a consequence, also in this class of models of hybrid dynamics, we have that, in order to extract information from the quantum component, we need a dissipative dynamics for the quantum system. Furthermore, the classical system acting as a measuring apparatus, as it receives this information, has necessarily a dissipative dynamics, because of the probabilistic nature of any quantum information.

Even without extraction of information from the quantum system, to know that there is a jump in the counting process $\Pi$ means to know that the quantum system has been subjected to a certain kind of jump. But for this we have to know before the dynamical laws of the system; we are not observing the quantum system, but a classical system which is acting on the quantum component.

\subsubsection{The non-linear SME and the conditional states.}
By the rules of stochastic calculus we can obtain the \emph{non-linear SME} satisfied by the normalized random states $\widehat \sigma_t$ defined in \eqref{defhatsigma}:
\begin{multline}
\rmd \langle \widehat\sigma_t,a \rangle= \langle \widehat\sigma_{t_-},\Lcal_t[a] \rangle \rmd t +\sum_{k=1}^d \langle \widehat \sigma_{t_-},{L_{kt}}^\dagger a +aL_{kt}-m_k(t)a\rangle \rmd \widehat  W_k(t)
\\ \label{nonlinearSME}{}+ \int_{\Eo_t}\left\{I_t(u)^{-1}\big\langle \widehat \sigma_{t_-}, \bigl(J_t(u)^\dagger +\openone\bigr)a \bigl(J_t(u)+\openone\bigr)\big\rangle -\langle \widehat \sigma_{t_-} ,a\rangle \right\}\widehat \Pi(\rmd u,\rmd t);
\end{multline}
$\Eo_t$ is the random subset of $\Uscr$, defined by $\Eo_t(\omega)=\{u\in \Uscr: I_t(u;\omega)\neq 0\}$. Equation \eqref{nonlinearSME} follows from \cite[Remark 3.6]{BarH95}.

The two coupled SDEs \eqref{XnewPa} and \eqref{nonlinearSME} have been obtained by construction. This proves that these equations admit a solution and that it exists a probability $\Pbb$ under which the processes $\widehat W_k$ are independent Wiener processes and $\Pi(\rmd u,\rmd t)$ is a point measure with random compensator $I_t(u)\nu(\rmd u)\rmd t$; the uniqueness remains open. In a simplified case, without jump contributions, uniqueness has been proved in \cite[Theor.\ 5.2]{BarG09}.

\begin{remark}[Conditional state]
Let $\Escr_t$\marginlabel{$\Escr_t$} be the $\sigma$-algebra generated by $X(r)$, $0\leq r\leq t$; note that it could depend on the choice of the initial condition, the random variable $X(0)$. Then, we can consider the quantum state conditioned on the observation of the process $X$ up to time $t$, the \emph{a posteriori state} \marginlabel{$\rho_t$} (or \emph{conditional state})
\begin{equation}\label{apos}
\rho_t:=\Ebb_\Pbb[\widehat \sigma_t|\Escr_t]\equiv \Ebb_\Pbb[\widehat \sigma_t|\Escr], \qquad \Escr= \bigvee_{t\geq 0} \Escr_t.
\end{equation}
\end{remark}

When the two filtrations coincide, $\Escr_t=\Fscr_t$, the conditional states satisfy the non-linear SME \eqref{nonlinearSME}. In general, the possibility of having a closed evolution equation for the conditional states depends on the structure of the coefficients; examples are given in \cite[Sec.\ 4]{BarH95}. When not all the components $X_1(t),\ldots, X_s(t)$ of the classical system are observed, one has to take as $\{\Escr_t\}$ the filtration generated by the observed components of $X(t)$; this is a case of partial observation and it opens to the application of all the notions and results of classical filtering theory \cite{Xi08}, which indeed is at the origin of the idea of observations in continuous time in quantum theories \cite{Bel88}.

\begin{remark}[Mean state]
When the classical component is not of interest and not observed, the quantum component is described by the \emph{a priori state}, the mean state \marginlabel{$\eta_t$}
\begin{equation}\label{eta}
\eta_t:=\Ebb_\Qbb[\sigma_t]=\Ebb_\Pbb[\widehat \sigma_t]=\Ebb_\Pbb[\rho_t].
\end{equation}
We can say that $\{\Pbb(\rmd \omega),\; \rho_t(\omega)\}$ is an \emph{ensemble} of quantum states; then, $\eta_t$ is the mean state of this ensemble. Given a mixed state $\eta_t$, there are infinitely many decompositions in ensembles of states; the decomposition in terms of conditional states has a physical meaning, as it is fixed by the observation of the classical component of the hybrid system.
\end{remark}

By taking the $\Qbb$-mean of the linear SME \eqref{linearSME} we get immediately
\begin{equation}\label{redqeq}
\frac{\rmd \ }{\rmd t}\, \eta_t= \Ebb_\Qbb\big[{\Lcal_t}_*[\sigma_{t_-}]\big]
\end{equation}
where the random Liouvillian ${\Lcal_t}_*$ is the pre-adjoint of the operator defined by \eqref{Lcalt}. When the coefficients in the definition of $\Lcal(x)$ do not depend on $x$, we get that ${\Lcal_t}_*$ is non random and the a priori state satisfies a closed quantum master equation. In general, equation \eqref{redqeq} is not closed and the reduced dynamics of the quantum component is not Markovian. The $x$-dependence in the Liouville operator \eqref{Lcalx} gives the dependence of the quantum dynamics on the classical component and it controls the transfer of information from the classical system to the quantum one. This $x$-dependence represents a feedback of the classical system on the quantum one; moreover, it could  appear in the Hamiltonian term alone: the information transfer from the classical to the quantum component is not related to dissipation.

The notions of conditional and mean states are typical of the approach to hybrid systems based on measurement and filtering \cite{Bel88,Hol01,ZolG97,BarG09,WisM10}, but the idea of considering the full dynamics and its reduced (or marginal) version appears also in other approaches \cite{Pomar+23}.

\subsection{Probabilities and instruments}\label{sec:SMEtoInstr}

A key point in the trajectory formulation is that the physical probabilities of Sec.\ \ref{sec:nprob} can be expressed by means of instruments, which are \emph{completely positive linear maps} (see Def.\ \ref{def:instr}); in this way the linearity of a quantum theory is respected. The SME \eqref{nonlinearSME} can not include arbitrary non-linearities: it is non-linear only due to normalization and the change of probability; moreover, it is connected to the linear SME \eqref{linearSME}.

\subsubsection{The stochastic map.}\label{sec:Lambda}
By varying the initial condition in \eqref{linearSME}, we have that the linear SME  defines a linear map on the trace class.
As in Proposition \ref{prop:clpr}, $X^x(t) $ is the solution of the stochastic equation \eqref{Xprocess} starting at $x$. Then, we consider the linear SME with the process $X^x(t) $ inside its coefficients \eqref{LtJt}; so, the random coefficients $\Lcal_t$, $L_{kt}$, $J_t(u)$ in \eqref{linearSME} are fixed and the equation is linear in the initial quantum state.

Let us denote by $\Lambda^{r,x}_{t\;*}[\rho]$, $t\geq r$, the solution of \eqref{linearSME} with $\rho$ as initial condition at time $r$; by construction, $\Lambda^{r,x}_{t\;*}$ is a random, completely positive, linear map on $\Tscr(\Hscr)$. Let $\Lambda^{r,x}_t$ \marginlabel{$\Lambda^{r,x}_t$} be the adjoint map on $\Bscr(\Hscr)$; its properties are collected in the following proposition.

\begin{proposition}\label{prop:Lambda} The random map $\Lambda^{r,x}_t$ is linear, completely positive, normal,   and it enjoys the composition property
\begin{equation}\label{propL1}
\Lambda^{t_1,x}_{t_2}\circ\Lambda^{t_2,x}_{t_3}=\Lambda^{t_1,x}_{t_3}, \qquad 0\leq t_1\leq t_2\leq t_3.
\end{equation}
We have also
\begin{equation}\label{propL2}
\Lambda^{t',x}_{t'+t} =  \Lambda^{0,X^x(t')}_{t}, \qquad \Tr \{\rho \Lambda^{r,x}_t[\openone]\}=1, \quad \forall \rho\in \Sscr(\Hscr).
\end{equation}
\end{proposition}
\begin{proof}
The first properties and the composition rule \eqref{propL1} follow from the construction of $\sigma_t$ in Sec.\ \ref{sec:sigmat} and from \eqref{linearSME}.  By the fact that $X$ is a Markov process under $\Qbb$ and that this process appears in the right hand side of \eqref{linearSME} only as $X(t_-)$, the first property in \eqref{propL2} follows. The second property in \eqref{propL2} follows from \eqref{def:pt}.
\end{proof}

If $\sigma_t$ is the solution of \eqref{linearSME} with the initial condition $\sigma_0$ given in \eqref{inst}, we have
\begin{equation}
\langle \sigma_t,a\rangle = \langle \sigma_0,\, \Lambda^{0,X(0)}_{t}[a]\rangle, \qquad \langle \sigma_{t'+t},a\rangle = \langle \sigma_{t'},\, \Lambda^{t',X(0)}_{t'+t}[a]\rangle=  \langle \sigma_{t'},\, \Lambda^{0,X(t')}_{t}[a]\rangle.
\end{equation}

\subsubsection{The transition instruments.}\label{sec:transinstr}
By analogy with the transition probabilities, it is possible to introduce instruments which depend on the initial value of the classical system, some kind of \emph{transition instruments} \cite{BarH95}. Indeed,
we can define the family of maps \marginlabel{$\Ical_t(\Eo|x)$}
\begin{equation}\label{I(E|x)}
\langle \rho,\Ical_t(\Eo|x)[a]\rangle = \Ebb_{\Qbb}\left[\langle \Lambda^{0,x}_{t\;*}[\rho], a\rangle \ind_\Eo\big(X^x(t)\big)\right], \qquad \forall \rho \in \Tscr(\Hscr), \qquad \forall a\in\Bscr(\Hscr),\qquad \forall\Eo\in \Bcal(\Rbb^s).
\end{equation}
By $\ind_\Eo$ we denote the \emph{indicator function} of a generic set $\Eo$: \marginlabel{$\ind_\Eo$}
\[
\ind_\Eo(x)=\begin{cases} 1 & \text{if } \ x\in \Eo,
\\
0 & \text{if } \ x\notin \Eo.\end{cases}
\]
Moreover, we shall denote by $\delta_x(\cdot)$ \marginlabel{$\delta_x(\cdot)$} the trivial measure concentrated in the point $x\in\Rbb^s$ (Dirac measure).

\begin{proposition}\label{prop:instr}
Equation \eqref{I(E|x)} defines an instrument $\Ical_t(\cdot|x)$ on the $\sigma$-algebra $\Bcal(\Rbb^s)$ (see Def.\ \ref{def:instr}). For $t=0$ we get the trivial instrument
\begin{equation}\label{I:t=0}
\Ical_{0}(\Eo|x)[a]=\delta_x(\Eo)\, a.
\end{equation}
Moreover, the family of instruments \eqref{I(E|x)} enjoys  the following composition property
\begin{equation}\label{qCKeq}
\Ical_{t+t'}(\Eo|x)=\int_{z\in \Rbb^s} \Ical_{t'}(\rmd z|x)\circ\Ical_t(\Eo|z).
\end{equation}
\end{proposition}
Equation \eqref{qCKeq} is a quantum analogue of the Chapman-Kolmogorov identity for transition probabilities \eqref{CKeq}.
\begin{proof}
By using the properties of the maps $\Lambda^{0,x}_{t}$, given in Proposition \ref{prop:Lambda}, it is possible to check that the maps defined in \eqref{I(E|x)} satisfy all the properties needed in the definition of instrument (Def.\ \ref{def:instr}).

The composition property \eqref{qCKeq} is proved by direct computations. By using \eqref{propL1} and \eqref{propL2}, we have
\begin{multline*}
\langle \rho, \Ical_{t+t'}(\Eo|x)[a]\rangle =\Ebb_\Qbb\left[ \Ebb_\Qbb\left[\langle \Lambda^{0,x}_{t'\;*}[\rho],\; \Lambda^{t',x}_{t'+t}[a]\rangle \ind_\Eo\big(X^x(t+t')\big)\Big|\Escr_{t'}\right]\right]
\\ {}=\int_{z\in \Rbb^s}\Ebb_\Qbb\left[ \Ebb_\Qbb\left[\langle \Lambda^{0,x}_{t'\;*}[\rho],\; \Lambda^{t',x}_{t'+t}[a]\rangle \ind_{\Eo-z}\big(X^x(t+t')-X^x(t')\big)\ind_{\rmd z}\big(X^x(t')\big)\Big|\Escr_{t'}\right]\right]
\\ {}=\int_{z\in \Rbb^s}\Ebb_\Qbb\left[ \langle \Lambda^{0,x}_{t'\;*}[\rho],\; \Ical_t(E|z)[a]\rangle \ind_{\rmd z}\big(X^x(t')\big)\right]
=\int_{z\in \Rbb^s} \langle \rho,\;\Ical_{t'}(\rmd z|x)\circ\Ical_t(\Eo|z)[a]\rangle;
\end{multline*}
$\ind_{\rmd z}\big(X^x(t')\big)$ is the measure-theoretic notation for $\delta \big(z-X^x(t')\big)\rmd z$; the quantity $X^x(t')$ is implicitly contained in $\Lambda^{0,x}_{t'\;*}$.
\end{proof}

The transition instruments \eqref{I(E|x)} are uniquely determined by their Fourier transform, the \emph{characteristic operator}: for $k\in \Rbb^s$,
\begin{equation}\label{G(k|x)}
\langle \rho,\Gcal_t(k|x)[a]\rangle =\int_{\Rbb^s}\rme^{\rmi k\cdot y}\langle \rho,\Ical_t(\rmd y|x)[a]\rangle = \Ebb_{\Qbb}\left[\langle \Lambda^{0,x}_{t\;*}[\rho], a\rangle \exp\left\{\rmi k\cdot X^x(t)\right\}\right].
\end{equation}

The ``Markovian'' character of the hybrid dynamics is expressed by the fact that the instruments contain only the length of the time interval (not initial and final times) and that equation \eqref{qCKeq} holds: the transition instruments are ``time-homogeneous'', as the classical transition functions for a time-homogeneous Markov process, compare the properties of instruments and equations \eqref{I(E|x)}--\eqref{qCKeq} with equations \eqref{transf}--\eqref{CKeq}.

Property \eqref{qCKeq} represents a compatibility condition among the various instruments at different times. Via Kolmogo\-rov's extension theorem \cite[Theor.\ 1.8]{Sato99}, this property allows to reconstruct the whole probability law of the process $X(t)$ starting from the joint probabilities at the times $0\leq t_1<t_2<\cdots<t_m$, given by
\begin{multline}
\Pbb[X(t_1)\in \Eo_1, X(t_2)\in \Eo_2, \ldots, X(t_m)\in \Eo_m]
=\Ebb_\Qbb\big[\langle \sigma_{t_m}, \openone\rangle \ind_{\Eo_1}\big(X(t_1)\big)\ind_{\Eo_2}\big(X(t_2)\big)\cdots \ind_{\Eo_m}\big(X(t_m)\big)\big]
\\  {}
=\int_{\Rbb^s}\Qbb_0(\rmd x) \big\langle \sigma_0(x),\; \int_{\Eo_1}\Ical_{t_1}(\rmd x_1|x)\circ \int_{\Eo_2}\Ical_{t_2-t_1}(\rmd x_2|x_1)\circ \cdots
\circ \int_{\Eo_m}\Ical_{t_m-t_{m-1}}(\rmd  x_m|x_{m-1})[\openone]\big\rangle, \label{genprobs}
\end{multline}
\[
\sigma_0(x)=\Ebb_\Qbb\big[\sigma_0\big|X(0)=x\big], \qquad \Qbb_0(\Eo):=\Qbb[X(0)\in \Eo];
\]
here $\Pbb$ is the probability defined in Sec.\ \ref{sec:nprob}, which depends on the initial conditions $\sigma_0$, $ X(0)$. Note that we have $\langle \sigma_{t_m}, \openone\rangle=\Tr\{\sigma_{t_m}\}=p_{t_m}$, which is the probability density for events up to the biggest time $t_m$, in agreement with the discussion following the definition \eqref{physprob} of $\Pbb$.

In the construction of the hybrid dynamics, we have used a  two-step procedure: firstly, the SDEs \eqref{Xprocess} and \eqref{linearSME} under the reference probability $\Qbb$ have been introduced, and, then, the SDEs \eqref{XnewPa} and \eqref{nonlinearSME} under the physical probability $\Pbb$.  By following this path, already developed in the field of time-continuous measurements \cite{BarG09}, we were able  to give rise to the instruments \eqref{I(E|x)} with the right properties, given in Proposition \ref{prop:instr}, and with the connection  \eqref{genprobs} with the physical probability $\Pbb$. In this way, the connection with the general formulation of a quantum theory is preserved.

\section{Quantum-classical dynamical semigroup}\label{sec:hds}

The hybrid dynamics has been introduced through the formalism of quantum trajectories; however, a Markovian hybrid dynamics is usually introduced by proposing suitable master equations. The most standard choice is to take probability densities (with respect to Lebesgue measure) as state space of the classical component \cite{DGS00,Dio14,Op+22,Dio23,Opp+23,Bar23,LOW22}, but this choice is too restrictive in our case, as shown in \ref{app:classM} for the pure classical case.

A more general approach is to construct a dynamical semigroup in a suitable ``space of observables'' \cite{DamWer23,BarW23} and to take for the classical component the spaces of functions $C_{\rm b}(\Rbb^s)$ and $B_{\rm b}(\Rbb^s)$, introduced in Sec.\ \ref{sec:notations}. Then, for the hybrid system we introduce \marginlabel{$\Ascr_1$}  the $C^*$-algebras \cite{DamWer23}:
\[
\Ascr_1:=\Bscr(\Hscr)\otimes C_{\rm b}(\Rbb^s)=C_{\rm b}\big(\Rbb^s;\Bscr(\Hscr)\big), \qquad \Ascr_2:=\Bscr(\Hscr)\otimes B_{\rm b}(\Rbb^s)=  B_{\rm b}\big(\Rbb^s;\Bscr(\Hscr)\big).
\]
The elements of $\Ascr_2$ \marginlabel{$\Ascr_2$} are bounded functions from $\Rbb^s$ to $\Bscr(\Hscr)$, while $\Ascr_1 \subset\Ascr_2$ is made up by the continuous bounded functions. The norm of $F\in \Ascr_i$ is given by $\norm F=\sup_{x\in \Rbb^s} \norm{F(x)}$, where $\norm{F(x)}$ is the norm in $\Bscr(\Hscr)$. The unit element $\Ibb$ in $\Acal_i$ is given by $\Ibb=\openone \otimes 1$;\marginlabel{$\Ibb$} we can also write $\Ibb(x)=\openone $.

\subsection{The hybrid semigroup}
By using the stochastic process $X^x(t)$, defined inside Proposition \ref{prop:clpr}, and the stochastic map $\Lambda^{0,x}_{t}$, defined in Sec.\ \ref{sec:Lambda}, we can define the linear map $\Tcal_t$, $t\geq 0$, on $\Acal_2$ by
\begin{equation}\label{def:Tt}
\langle \rho, \Tcal_t[F](x)\rangle =\Ebb_{\Qbb}\left[\big\langle \rho,\, \Lambda^{0,x}_{t}\big[F\big(X^x(t)\big)\big]\big\rangle \right], \qquad \forall \rho \in \Tscr(\Hscr), \quad  \forall F \in  \Acal_2, \quad \forall x\in\Rbb^s.
\end{equation}
As we shall see $\Tcal_t$ is a contraction; so, it is enough to give $\Tcal_t$  on the product elements and then to extend it by linearity and continuity. By using the instruments \eqref{I(E|x)}, we can also write
\begin{equation}\label{T+instr}
\langle \rho, \Tcal_t[a\otimes f](x)\rangle=\int_{y\in\Rbb^s} f(y) \langle \rho, \Ical_t(\rmd y|x)[a]\rangle,
\qquad
\rho\in \Tscr(\Hscr), \quad a\in \Bscr(\Hscr), \quad f\in B_{\rm b}(\Rbb^s), \quad x\in\Rbb^s.
\end{equation}
On the other hand, given $\Tcal_t$, we can get the instruments \eqref{I(E|x)} and their characteristic operator \eqref{G(k|x)} by
\begin{equation}\label{TtoI}
\langle \rho, \Ical_t(\Eo|x)[a]\rangle=\langle \rho, \Tcal_t[a\otimes \ind_\Eo](x)\rangle , \quad
\langle \rho,\Gcal_t(k|x)[a]\rangle  =\langle \rho, \Tcal_t[a\otimes f_k]\rangle, \quad f_k(x)=\rme^{\rmi k\cdot x}.
\end{equation}
To get the characteristic  operators it is enough to have $\Tcal_t$ defined on $\Ascr_1$, while to get directly the instruments we need $\Tcal_t$ defined on $\Ascr_2$. As $\Tcal_t$, $t\geq 0$, turns out to be a semigroup, we call it the \emph{hybrid dynamical semigroup } associated with the stochastic dynamics \eqref{Xprocess}, \eqref{linearSME}.

\begin{proposition}\label{prop:Tt}
The map $\Tcal_t$ is linear and completely positive, which means: for any choice of the integer $N\geq 1$, one has
\begin{equation}
\sum_{i,j=1}^N \langle \psi_j| \Tcal_t[F_j^*F_i](x)\psi_i\rangle\geq 0, \qquad \forall F_j\in \Acal_2, \quad \psi_j\in \Hscr, \quad \forall x\in\Rbb^s.
\end{equation}
We have also
\begin{equation}\label{3props}
\Tcal_0=\id,
\qquad
\Tcal_t[\Ibb]=\Ibb,\qquad \norm{\Tcal_t[F]}\leq \norm{F},\qquad \Tcal_t:\Acal_2 \mapsto \Acal_2 ,
\end{equation}
\begin{equation}\label{lastprop}
\Tcal_t\circ\Tcal_r=\Tcal_{t+r} .
\end{equation}
\end{proposition}
\begin{proof}
Linearity and complete positivity follow from the same properties of the integration operation $\Ebb_\Qbb$ and of the stochastic map $\Lambda^{0,x}_{t}$ (Prop.\ \ref{prop:Lambda}).
Then, $\Tcal_0=\id$ follows from \eqref{I:t=0} and \eqref{T+instr};
$\Tcal_t[\Ibb]=\Ibb$ follows from \eqref{T+instr} and the normalization of the instruments (Def.\ \ref{def:instr}).

By using $\norm{F}^2\openone -F^*(x)F(x)\geq 0$, the positivity of $\Tcal_t$, and the preservation of the unit, we get
\[
\Tcal_t[F^*F](x)\leq \norm{F}^2 \openone.
\]
By 2-positivity, linearity, and the preservation of the unit, one gets a Kadison-like inequality \cite[Sec.\ 3.1.1]{Hol01}, \cite[Ex.\ 6.7 p.\ 107]{Hol12}:
\[
\Tcal_t[F^*F](x)\geq \Tcal_t[F^*](x)\Tcal_t[F](x)\equiv\Tcal_t[F](x)^*\Tcal_t[F](x) .
\]
By combining the two inequalities we get $\Tcal_t[F]^*\Tcal_t[F]\leq \norm F ^2 \Ibb$; this gives the contraction property in \eqref{3props} and the fact that $\Tcal_t$ maps $\Acal_2$ into itself.

By using conditional expectations, we have
\begin{multline*}
\langle \rho , \Tcal_{t+r}[a\otimes f](x)\rangle =\Ebb_{\Qbb}\left[\langle \Lambda^{0,x}_{t\;*}[\rho], \Lambda^{t,x}_{t+r}[a]\rangle f\big(X^x(t+r)\big)\right]
\\ {}=\Ebb_{\Qbb}\left[\big\langle \Lambda^{0,x}_{t\;*}[\rho], \Ebb_\Qbb\left[\Lambda^{t,x}_{t+r}[a]f\big(X^x(t+r)\big)\big|\Escr_t,X^x(t)\right]\big\rangle \right].
\end{multline*}
By the Markov property of the process $X^x$, we get
\[
\Ebb_\Qbb\left[\Lambda^{t,x}_{t+r}[a]f\big(X^x(t+r)\big)\big|\Escr_t,X^x(t)\right]=\Tcal_r[a\otimes f]\big(X^x(t)\big);
\]
so, we obtain the composition property:
\[
\langle \rho , \Tcal_{t+r}[a\otimes f](x)\rangle= \Ebb_{\Qbb}\left[\big\langle \Lambda^{0,x}_{t\;*}[\rho], \Tcal_r[a\otimes f]\big(X^x(t)\big)\big\rangle \right]
=\langle \rho, \Tcal_t\circ\Tcal_r[a\otimes f](x)\rangle.
\]
This gives the composition property \eqref{lastprop}.
\end{proof}

\begin{remark}[Convolution semigroup of instruments] When all the coefficients introduced in  Assumptions \ref{ass:classical} and \ref{Ass:1} do not depend on $x\in \Rbb^s$, from \eqref{Xprocess}, \eqref{linearSME}, \eqref{I(E|x)}, \eqref{qCKeq}, \eqref{G(k|x)} we obtain
\[
X^x(t)= x+ X^0(t), \qquad \Ical_t(E|x)= \Ical_t(E-x|0), \qquad \Gcal_{t+t'}(k|0)=\Gcal_{t'}(k|0)\circ \Gcal_t(k|0).
\]
Under these conditions the associated dynamical semigroup \eqref{T+instr} has been fully studied under the name of \emph{convolution semigroup of instruments} and generalized to cases with a classical component not only in $\Rbb^s$ \cite{Hol86,Hol88QPIII,Hol89,Bar89,BarL91,BarHL93,BarP96,Maa23}. In this construction a key point has been to exploit the analogies with the classical infinitely divisible distributions \cite{Sato99,App09}.
\end{remark}

The dual $C^*$-algebra of $\Ascr_2$ is the space containing the hybrid states and the adjoint semigroup ${\Tcal_t}^*$ gives the dynamics of the hybrid states; therefore, the hybrid dynamics is completely positive and linear. Another relevant point is Eq.\ \eqref{TtoI}, which defines the transition instruments $\Ical_t(\cdot|x)$ once a semigroup $\Tcal_t$, satisfying the properties in Proposition \ref{prop:Tt}, is given. Therefore, the dynamics ${\Tcal_t}^*$ gives not only the hybrid state at the time $t$, but also all the multi-time probabilities \eqref{genprobs}. Once again, to give a hybrid dynamical semigroup (linear and completely positive) allows to respect the general structure of a quantum theory, as discussed at the end of Sec.\ \ref{sec:transinstr} in the case of the stochastic formulation.

\subsection{The generator}

We have not studied the continuity properties in time of the dynamical semigroup \eqref{def:Tt}; however, by using stochastic calculus we can see that it is differentiable when applied to the elements of $\Ascr_3:=C^2_{\rm b}\big(\Rbb^s;\Bscr(\Hscr)\big)\subset \Ascr_1$, \marginlabel{$\Ascr_3$} the space of the operator valued functions, which are two-times differentiable and have bounded and continuous derivatives.

\begin{proposition}
The hybrid dynamical semigroup  $\Tcal_t$ can be differentiated in weak sense
\begin{equation}\label{derT}
\frac{\rmd \ }{\rmd t}\,\langle \rho, \Tcal_t[F](x)\rangle =\langle \rho, \left(\Tcal_t\circ\Kcal\right)[F](x)\rangle, \qquad \forall \rho \in \Tscr(\Hscr), \quad  \forall F \in  \Ascr_3, \quad \forall x\in\Rbb^s.
\end{equation}
The formal generator $\Kcal$ is given by
\begin{multline}\label{Kcal}
\Kcal[a\otimes f](x) = f(x)\Lcal(x)[a]
+ a \sum_{i=1}^s \frac{\partial f(x)}{\partial x_i}\, c_i (x)+a\sum_{i,j=1}^s\frac 12\, \frac{\partial^2 f(x)}{\partial x_i\partial x_j}\sum_{k=1}^{d}b_{i}^k(x) b_{j}^k(x)
\\ +\sum_{i=1}^s \frac{\partial f(x)}{\partial x_i}\sum_{k=1}^{d}b_{i}^k(x)\left(a L_k(x)+L_k(x)^\dagger a \right)
\\ + \int_{\Uscr}\biggr(\left(f\big(x+g(x,u)\big) -f(x)\right) \bigl(J(x,u)^\dagger +\openone\bigr)a\bigl( J(x,u) +\openone\bigr)
-a\ind_{\mathsf{U}_0}(u)\sum_{i=1}^s \frac{\partial f(x)}{\partial x_i}\, g_i(x,u)\biggl)  \nu(\rmd u);
\end{multline}
$\Lcal(x)$ is defined in \eqref{Lcalx}.
\end{proposition}
\begin{proof}
By \eqref{linearSME}, \eqref{dfX}, and stochastic calculus, we get
\begin{multline*}
\rmd  \left(f\big(X(t)\big)\langle\sigma_t,a\rangle\right)\Big|_{X(t_-)=x} = f(x)\biggl(
\langle \sigma_{t_-},\Lcal(x)[a] \rangle \rmd t +\sum_{k=1}^d \langle \sigma_{t_-},{L_{k}(x)}^\dagger a +aL_{k}(x)\rangle \rmd W_k(t)
\\ {}+ \int_{\Uscr}\langle \sigma_{t_-}, J(x,u)^\dagger a J(x,u)+ J(x,u)^\dagger a+a J(x,u)\rangle \widetilde \Pi(\rmd u,\rmd t)\biggr)
\\ {}+ \langle\sigma_{t_-},a\rangle\Biggl( \sum_{i=1}^s \frac{\partial f(x)}{\partial x_i}\left( c_i (x)\rmd t +\sum_{k=1}^{d}b_{i}^k(x) \rmd W_k(t)\right)+\frac 12 \sum_{i,j=1}^s \frac{\partial^2 f(x)}{\partial x_i\partial x_j}\sum_{k=1}^{d}b_{i}^k(x) b_{j}^k(x)\rmd t
\\ {} +\rmd t \int_{\mathsf{U}_0}\left(f\big(x+g(x,u)\big) -f(x)-\sum_{i=1}^s \frac{\partial f(x)}{\partial x_i}\, g_i(x,u)\right)\nu(\rmd u)
\\ {}+\int_{\mathsf{U}_0}\left(f\big(x+g(x,u)\big) -f(x)\right)\widetilde \Pi(\rmd u ,\rmd t) +\int_{\Uscr\setminus \mathsf{U}_0}\left(f\big(x+g(x,u)\big) -f(x)\right) \Pi(\rmd u ,\rmd t)\Biggr)
\\ {}+\sum_{i=1}^s \frac{\partial f(x)}{\partial x_i}\sum_{k=1}^{d}b_{i}^k(x)\langle\sigma_{t_-},\,a L_k(x)+L_k(x)^\dagger a \rangle \rmd t
\\ {} + \int_{\Uscr}\left(f\big(x+g(x,u)\big) -f(x)\right)\langle \sigma_{t_-}, J(x,u)^\dagger a J(x,u)+ J(x,u)^\dagger a+a J(x,u)\rangle  \Pi(\rmd u,\rmd t).
\end{multline*}
Then, by taking the mean we have
\begin{multline*}
\Kcal[a\otimes f](x) = f(x)\Lcal(x)[a]
+ a\Biggl( \sum_{i=1}^s \frac{\partial f(x)}{\partial x_i} \,c_i (x)+\frac 12 \sum_{i,j=1}^s \frac{\partial^2 f(x)}{\partial x_i\partial x_j}\sum_{k=1}^{d}b_{i}^k(x) b_{j}^k(x)
\\ {} +\int_\Uscr\left(f\big(x+g(x,u)\big) -f(x)-\ind_{\mathsf{U}_0}(u)\sum_{i=1}^s \frac{\partial f(x)}{\partial x_i}\, g_i(x,u)\right)\nu(\rmd u)
\Biggr)
\\ {}+\sum_{i=1}^s \frac{\partial f(x)}{\partial x_i}\sum_{k=1}^{d}b_{i}^k(x)\left(a L_k(x)+L_k(x)^\dagger a \right)
\\ {}\quad{} + \int_{\Uscr}\left(f\big(x+g(x,u)\big) -f(x)\right)\left( J(x,u)^\dagger a J(x,u)+ J(x,u)^\dagger a+a J(x,u)\right)  \nu(\rmd u),
\end{multline*}
and this gives \eqref{Kcal}. The semigroup property gives the derivative at any time $t$ \eqref{derT}.
\end{proof}

Here we have not studied the problem of determining the dynamical semigroup $\Tcal_t$ from the generator \eqref{Kcal}. By construction $\Tcal_t$, defined by \eqref{def:Tt}, is a solution of \eqref{derT}; what is not proved is the uniqueness of this solution. The dynamics of the hybrid system is fully determined by the stochastic differential equations \eqref{Xprocess} and \eqref{linearSME} under Assumptions \ref{Ass:processes}, \ref{ass:classical}, \ref{Ass:1} (see Proposition \ref{prop:clpr}, Theorem \ref{theor:sigmaprop}).

\subsection{A quasi-free dynamics}\label{sec:qfdyn}
In \cite{BarW23} the structure of the most general quasi-free hybrid dynamical semigroup has been obtained. The ``quasi-free'' requirement means to map hybrid Weyl operators into a quantity proportional to hybrid Weyl operators \cite{DVV79}. In this case, the formal generator involves unbounded operators on the Hilbert space of the quantum component; so, the semigroup found in \cite{BarW23} is not in the class studied here. However, the comparison of the two formal generators gives hints to construct the stochastic formulation of the quasi-free  dynamics.

Firstly, let us take
\begin{gather*}
\Uscr=\Xi_1\oplus \Xi_0, \qquad \Xi_1=\Rbb^{2n}, \qquad \Xi_0=\Rbb^s, \qquad \mathsf{U}_0=\Sbb=\{u\in \Uscr: \abs u\leq 1\},
\\
c_i (x)=\alpha^0_i+ \sum_{j=1}^s x_j Z^{00}_{ji},\quad \alpha^0_i\in \Rbb \quad Z^{00}_{ji}\in\Rbb, \quad i=1,\ldots, s,
\\
b_{j}^k(x)=b_{j}^k\in\Rbb, \qquad g(x,u)=g(u)=P_0u \quad g_i(u)= u_{2n+i},
\end{gather*}
where $P_0$ and $P_1$ are the orthogonal projections $P_i\Uscr=\Xi_i $, $i=1,2$. With these choices
Assumptions \ref{Ass:processes} and \ref{ass:classical} hold, and  Eq.\ \eqref{Xprocess} becomes
\begin{multline}
\rmd X_i(t) = \alpha^0_i\rmd t + \sum_{j=1}^s X_j(t_-) Z^{00}_{ji}\rmd t +\sum_{k=1}^{d}b_{i}^k \rmd W_k(t)
\\ {}+\int_{\Sbb} g_{i}(u) \widetilde\Pi(\rmd u,\rmd t)
+ \int_{\Uscr\setminus \Sbb} g_{i}(u) \Pi(\rmd u,\rmd t), \qquad i=1,\ldots, s.
\end{multline}

Then, we represent the quantum component in the Hilbert space $\Hscr=L^2(\Rbb^n) $; by $\hat q_j$ and $\hat p_j$ we denote the usual position and momentum operators. We collect these operators in the vector $R^\T = (\hat p_1, \ldots,\hat p_n,- \hat q_1, \ldots ,-\hat q_n)$, where $\T$ means transposition; moreover, we introduce the usual Weyl operators
\[
W_1(P_1u)=\exp\left\{\rmi R^\T P_1u \right\} = \exp\left\{\rmi \sum_{j=1}^n \left(u_j\hat p_j-u_{n+j}\hat q_j\right)\right\}.
\]
For the operators appearing in Assumption \ref{Ass:1} we take
\begin{subequations}
\begin{gather}\label{qfJ}
J(x,u)=J(u)=- W_1(P_1u)-\openone,
\\ \label{qfLk}
L_k(x)=L_k=-\rmi\sum_{l=1}^{2n}R_ld_{l}^k=\rmi \sum_{l=1}^n \left(d_{n+l}^k \hat q_l -d_l^k \hat p_l\right), \qquad d_l^k\in \Cbb,
\\
H(x)=H_{\rm q}+H_x
+\int_{\Uscr}\biggl(\frac {\rmi}2 \left(-W_1( P_1u) +W_1( P_1u)^\dagger\right)-\ind_{\Sbb} (P_1u)^\T  R \biggr) \nu(\rmd u),
\\
H_{\rm q}=\beta^\T R+\frac 12 \,R^\T P_1 \sigma D^{11} \sigma^\T P_1 R,\quad \beta\in \Xi_1,
\\
H_x= x^\T P_0 ZP_1 R=\sum_{i=1}^s \sum_{j=1}^{2n} x_i Z^{01}_{ij}R_j.
\end{gather}
\end{subequations}
With these choices we have that the formal generator \eqref{Kcal} becomes
\begin{multline}
\Kcal[a\otimes f](x) = f(x)\Big\{\rmi[H_{\rm q}+H_x,a]+\sum_{k=1}^d\sum_{r,l=1}^{2n}d^k_l\overline{d_r^k}\left( (\sigma R)_r a (\sigma R)_l-\frac 12 \left\{ (\sigma R)_r  (\sigma R)_l,a\right\}\right)
\\ {}+ a \sum_{i=1}^s \left(\alpha^0_i+ \sum_{j=1}^s x_j Z^{00}_{ji}\right)\frac{\partial f(x)}{\partial x_i}+a\sum_{i,j=1}^s\frac 12\, \frac{\partial^2 f(x)}{\partial x_i\partial x_j}\sum_{k=1}^{d}b_{i}^k b_{j}^k
\\ {}+\sum_{i=1}^s \frac{\partial f(x)}{\partial x_i}\sum_{k=1}^{d}b_{i}^k\sum_{l=1}^{2n}\left(\rmi (\sigma R)_l \overline{d_l^k}a-\rmi a (\sigma R)_l d_l^k\right)
+ \int_{\Uscr}\biggl(f\big(x+g(u)\big)   W_1(\sigma P_1u)^\dagger a  W_1(\sigma P_1u) -af(x)
\\ {}-\ind_{\Sbb}(u)\rmi f(x) \left[(P_1u)\T\sigma R,a\right]-\ind_{\Sbb}(u)a\sum_{i=1}^s \frac{\partial f(x)}{\partial x_i}\, g_i(u) \biggr) \nu(\rmd u);
\end{multline}
then, it is possible to check that it coincides with the formal generator given in \cite[Proposition 5]{BarW23}.

However, the operators $L_k$ and $H(x)$ are unbounded, so that the construction of Secs.\ \ref{sec:SSE}, \ref{sec:sigmat} becomes purely formal.
In any case \eqref{linearSME}, \eqref{Lcalx} take the form
\begin{multline}
\rmd \langle \sigma_t,a \rangle=
\langle \sigma_{t_-},\Lcal_t[a] \rangle \rmd t +\rmi \sum_{k=1}^d \sum_{l=1}^{2n}\langle \sigma_{t_-},R_l\overline{d^k_l} a - a R_ld^k_l\rangle \rmd W_k(t)
\\ \label{SMEqf} {}+ \int_{\Uscr}\langle \sigma_{t_-}, W_1(P_1u)^\dagger a W_1( P_1u)-a\rangle \widetilde \Pi(\rmd u,\rmd t),
\end{multline}
\begin{multline}
\Lcal(x)[a]
=\rmi[H_{\rm q}+H_x,a]+\sum_{k=1}^d\sum_{r,l=1}^{2n}d^k_l\overline{d_r^k}\left( R_r a R_l-\frac 12 \left\{  R_r   R_l,a\right\}\right)
\\ {}+\int_\Uscr \left(W_1( P_1u)^\dagger a W_1( P_1u)-a -\rmi \ind_\Sbb(u)\left[(P_1u)^\T R,a\right]\right) \nu(\rmd u).
\end{multline}
Due to the fact that Assumption \ref{Ass:1} does not hold, the problem of existence and uniqueness of the solution of \eqref{SMEqf} is open; however, the identification of the two generators has given an explicit stochastic equation to study.

By \eqref{defmk}, \eqref{qfJ}, \eqref{qfLk} we get
\[
I_t(u)=1, \qquad m_{k}(t)= -2 \IM \sum_{l=1}^{2n}d_{l}^k\langle \widehat \sigma_{t_-} ,  R_l \rangle.
\]
The first equality implies that, under the physical probability $\Pbb$, $\Pi(\rmd u,\rmd t)$ remains a Poisson measure with the same compensator. In the second equality the unbounded operators $R_l$ appear; so, to have $m_k(t)$ well defined depends on the choice of the coefficients $d_{l}^k$ and of the initial state. When $\abs{m_{k}(t)}<+\infty$, from Eqs.\ \eqref{XnewPa}, \eqref{XnewP}, \eqref{nonlinearSME} we get the stochastic equation for the classical component
\begin{multline}
\rmd X_i(t) =\left(\alpha^0_i\rmd t + \sum_{j=1}^s X_j(t_-) Z^{00}_{ji}\rmd t +\sum_{k=1}^{d}b_{i}^k m_k(t)\right)\rmd t
+\sum_{k=1}^{d}b_{i}^k\rmd \widehat W_k(t)
\\ {}+\int_{\mathsf{U}_0} g_{i}(u) \widetilde\Pi(\rmd u,\rmd t)
+ \int_{\Uscr\setminus \mathsf{U}_0} g_{i}(u) \Pi(\rmd u,\rmd t),
\end{multline}
and the non-linear SME
\begin{multline}
\rmd \langle \widehat\sigma_t,a \rangle= \langle \widehat\sigma_{t_-},\Lcal_t[a] \rangle \rmd t+\rmi \sum_{k=1}^d \sum_{l=1}^{2n}\langle \widehat \sigma_{t_-},R_l\overline{d^k_l} a - a R_ld^k_l\rangle \rmd\widehat W_k(t)
\\ {}-\sum_{k=1}^d  m_k(t)\langle\widehat \sigma_{t_-},a\rangle \rmd \widehat  W_k(t)
+  \int_{\Uscr}\langle \widehat\sigma_{t_-}, W_1( P_1u)^\dagger a W_1( P_1u)-a\rangle \widetilde \Pi(\rmd u,\rmd t).
\end{multline}

In \cite{BarW23} it has been shown also that the adjoint of a quasi-free dynamical semigroup leaves invariant the sector of the state space based on classical densities with respect to Lebesgue measure, a fact which could not hold in the general case as discussed in \ref{app:classM}. So, the ``quasi-free'' requirement allows for unbounded quantum operators, but at the end it is a strong restriction on the possible structures.

\subsection{The master equation for hybrid states.} To have the explicit expression \eqref{Kcal} of the formal generator can help in clarifying the various physical interactions involved in the hybrid dynamics and it allows to compare the trajectory approach with the approaches based on master equations. Let us exclude the cases discussed at the end of \ref{app:classM} and assume that the pre-adjoint ${\Tcal_t}_*$ is well defined and that it leaves invariant $\Tscr(\Hscr)\otimes L^1(\Rbb^s)$. Then, we can take as hybrid states the $x$-dependent trace class operators $\varrho(x)$ such that $\varrho(x)\geq 0$ and $\Tr\{\varrho(x)\}$ is a probability density with respect to the Lebesgue measure; then, we set $\varrho_t(x)={\Tcal_t}_*[\varrho_0](x)$. To get the explicit master equation for $\varrho_t(x)$,  we assume a smooth $x$-dependence and we simplify the jump contributions by taking $g(x,u)=g(u)$.
By time-differentiation and using \eqref{Kcal} and \eqref{Lcalx}, we get the (heuristic) hybrid master equation
\begin{equation}\label{hME}
\frac {\rmd \ }{\rmd t}\,\varrho_t(x) =\Kcal_*[\varrho_t](x) = \Kcal_0[\varrho_t](x)+ \Kcal_{\rm diff}[\varrho_t](x) + \Kcal_{\rm jump}[\varrho_t](x),
\end{equation}
\begin{subequations}
\begin{equation}
\Kcal_0[\varrho_t](x)=-\rmi[H(x),\varrho_t(x)]- \sum_{i=1}^s \frac{\partial \ }{\partial x_i}\bigl( c_i (x)\rho_t(x)\bigr),
\end{equation}
\begin{multline}
\Kcal_{\rm diff}[\varrho_t](x)=\sum_{k=1}^{d}\Bigg( \frac 12\sum_{i,j=1}^s \frac{\partial^2 \ \ }{\partial x_i\partial x_j}\left(b_{i}^k(x) b_{j}^k(x)\rho_t(x)\right)-\frac 12 \left\{L_{k}(x)^\dagger  L_{k}(x), \varrho_t(x) \right\}
\\  {}+ L_{k}(x) \varrho_t(x) L_{k}(x)^\dagger-\sum_{i=1}^s \frac{\partial \ }{\partial x_i}\Bigl(b_{i}^k(x)\bigl( L_k(x)\varrho_t(x)+\varrho_t(x)L_k(x)^\dagger  \bigr)\Bigr)\Bigg),
\end{multline}
\begin{multline}
\Kcal_{\rm jump}[\varrho_t](x) =\int_{\Uscr}\nu(\rmd u)\Biggl(
 \Bigl(J\big(x-g(u),u\big) +\openone\Bigr)\varrho_t\big(x-g(u)\big)\Bigl( J\big(x-g(u),u\big)^\dagger +\openone\Bigr)
\\ {}-\frac 12 \left\{\Bigl(J(x,u)^\dagger +\openone\Bigr)\Bigl( J(x,u)+\openone\Bigr),\varrho_t(x)\right\}
+\frac 12 \left[J(x,u)^\dagger- J(x,u), \varrho_(x)\right]+\ind_{\mathsf{U}_0}(u)\sum_{i=1}^s g_i(u)\,\frac{\partial \varrho_t(x)}{\partial x_i}  \Biggr).
\end{multline}
\end{subequations}

This structure allows to compare, at a heuristic level, the present approach  to other approaches based on master equations for hybrid states, so giving an idea of possible physical applications. References \cite{Op+22,Opp+23,Opp+23grav,OppWD22} introduce master equations of jump type which are particular cases of \eqref{hME} with $\Kcal_{\rm diff}=0$. Then, the diffusive case is usually introduced by a suitable limit un the jump case; references \cite{LOW22,Op+22,Opp+23,Opp+23grav,OppWD22,Dio23} give master equations of diffusive type which are particular cases of \eqref{hME} with $\Kcal_{\rm jump}=0$. Examples with two-level systems and harmonic oscillators are constructed. The most relevant application is the attempt of connecting gravity and quantum matter in the so called ``Newtonian limit'' \cite{OppWD22,Opp+23grav}.

\section{Examples} \label{sec:ex} In this section we collect some simple examples of hybrid dynamics, just to illustrate the possibilities of the general approach, developed in the previous sections.

\subsection{A purely classical case}
As discussed in Sec.\ \ref{sec:nprob}, the hybrid dynamics can produce a change of the underling probability law. This can happen even in a purely classical case.

Let us take $\Hscr=\Cbb$, which means that no quantum component is present. Now the quantities introduced in Assumption \ref{Ass:1} become complex functions; conditions \eqref{somecond} continue to hold for the functions $L_k(x)$ and $J(x,u)$. In this ``classical'' case we get from \eqref{LtJt}, \eqref{defmk}
\begin{equation*}
m_k(t)=2\RE L_k\big(X(t_-)\big), \qquad I_t(u)= \abs{J\big(X(t_-),u\big)+1}^2.
\end{equation*}
Then, there is no dynamics in the quantum component, as \eqref{nonlinearSME} gives $\rmd \widehat\sigma_t=0$, while Eq.\ \eqref{linearSME} reduces to \eqref{eq:pt}. So, the probability density $p_t$ has the non trivial evolution \eqref{eq:pt} and it produces   the new probability $\Pbb$ \eqref{physprob}, even without interaction with a quantum component.

Under $\Pbb$, the classical process $X(t)$ satisfies the stochastic differential equation \eqref{XnewPa}, which looks like very similar to the original stochastic equation \eqref{Xprocess} for a Markov process, apart from a peculiar difference:
under $\Pbb$ the jump noise $\Pi(\rmd u,\rmd t)$ is no more a Poisson measure, but a point process whose law depends on the solution $X(t)$ of the stochastic equation itself, as discussed before \eqref{hatPi}. The advantage of working under the reference probability $\Qbb$ is to avoid the mathematical problems related to the involved stochastic equation \eqref{XnewPa}; the construction of the classical component $X(t)$ is split into the two standard stochastic differential equations \eqref{Xprocess} and \eqref{eq:pt}.

The effect of a non trivial probability density is apparent also from the formal generator \eqref{Kcal}, which now is purely of classical type:
\begin{multline*}
\Kcal[f](x) = \Biggl\{
\sum_{i=1}^s \frac{\partial f(x)}{\partial x_i}\left(c_i(x)+2\sum_{k=1}^{d}b_{i}^k(x)\RE L_k(x)\right)
+\sum_{i,j=1}^s\frac 12\, \frac{\partial^2 f(x)}{\partial x_i\partial x_j}\sum_{k=1}^{d}b_{i}^k(x) b_{j}^k(x)
\\ {} + \int_{\Uscr}\left(\left(f\big(x+g(x,u)\big) -f(x)\right) \abs{ J(x,u) +1}^2 -\ind_{\mathsf{U}_0}(u)\sum_{i=1}^s \frac{\partial f(x)}{\partial x_i}\, g_i(x,u)\right)  \nu(\rmd u)\Biggr\}.
\end{multline*}
By comparing it with the formal generator \eqref{Kcl}, we see immediately the corrections due to the change of probability: $c_i(x) \to c_i(x)+2\sum_{k=1}^{d}b_{i}^k(x)\RE L_k(x)$ in the first term and $1 \to \abs{ J(x,u) +1}^2$ in the jump term. While the first modification is only a redefinition of $c_i(x)$, the second modification is non trivial. Once again, the last modification is  related to the fact that $\Pi(\rmd u,\rmd t)$  is no more a Poisson process, but a counting process with a random compensator $\abs{J\big(X(t_-),u\big)+1}^2\nu(\rmd u)\rmd t$.

To take a non trivial quantum component ($\Hscr\neq \Cbb$), but with all the coefficients in Assumption \ref{Ass:1} proportional to the identity $\openone$, would give the same result on the classical component; indeed, this choice of the coefficients implies that the classical component does not receive any back-action from the quantum component. Also intermediate situations can arise with suitable choices of the coefficients.
This example shows that a change of probability is not a purely quantum effect, but it may be due to the reaction of the quantum component on the classical one and to auto-interactions in the classical component.

\subsection{Jumps of discrete type}\label{sec:Udiscr}

In the classical case, described by \eqref{dfX}, the integral over $\Uscr$ has the role of giving rise to a continuity of types of jumps; moreover, the integral over $\Uo_0$ allows for `infinitely many small jumps'. These features are inherited by the hybrid case of Sec.\ \ref{sec:SDE}. A simpler case, with sums instead of integrals, can help in the interpretation of the various terms.

In this section we consider the case of discrete types of jumps, when $\Uscr$ is a discrete and finite set:
\[
\Uscr=1, \ldots ,l, \qquad \mathsf{U}_0 =\emptyset, \qquad \nu(\{k\})=\lambda_k, \qquad \Pi(\{k\},\rmd t)=\rmd N_k(t).
\]
In this case it is convenient a change of notation to simplify the expressions $J(x,u)+\openone$ appearing in the generator \eqref{Kcal} and in the SMEs \eqref{linearSME}, \eqref{nonlinearSME}. To get this, we make the replacements
\[
J(x,k) \to J_k(x)-\openone, \qquad H(x)\to H(x)-\frac \rmi 2 \sum_{k=1}^l\lambda_k\left( J_k(x)^\dagger -J_k(x)\right).
\]
The pre-adjoint of the Liouville operator \eqref{Lcalx}, acting on the trace class,  becomes
\begin{multline}
\Lcal(x)_*[\rho]=
-\rmi[H(x),\rho]+\sum_{k=1}^d\left( L_{k}(x) \rho L_{k}(x)^\dagger-\frac 12 \left\{L_{k}(x)^\dagger  L_{k}(x),\rho\right\}\right)
\\ \label{Lcal*}{}+\sum_{k=1}^l\lambda_k\left(J_k(x)\rho J_k(x)^\dagger -\frac 12 \left\{J_k(x)^\dagger  J_k(x),\rho\right\}\right).
\end{multline}
Moreover, as in \eqref{Lcalt}, \eqref{defmk} we introduce $m_{k}(t)$ and we set
\begin{equation}\label{LJI}
{\Lcal_t}_*=\Lcal\big(X(t_-)\big)_*,\qquad J_{kt}=J_k\big(X(t_-)\big),\qquad I_{kt}=\langle \widehat \sigma_{t_-},\, J_{kt}^\dagger J_{kt}\rangle,
\end{equation}

Under the reference probability $\Qbb$, $N_k(t), \; k=1,\ldots ,l,$ are independent Poisson processes of intensities $\lambda_k$; we set $\rmd  \widetilde N_k(t)=\rmd N_k(t)-\lambda_k\rmd t$. Then, the hybrid dynamics is given by the coupled SDEs  \eqref{Xprocess}, \eqref{linearSME},  which now become
\begin{equation}
\rmd X_i(t) =c_i \big(X(t_-)\big)\rmd t +\sum_{k=1}^{d}b_{i}^k\big(X(t_-)\big) \rmd W_k(t)
+\sum_{k=1}^l g_{i}\big(X(t_-),k\big) \rmd N_k(t),
\end{equation}
\begin{equation}\label{klSME}
\rmd \sigma_t=  {\Lcal_t}_*[\sigma_{t_-}]  \rmd t +\sum_{k=1}^d\left( L_{kt}\sigma_{t_-}+\sigma_{t_-}{L_{kt}}^\dagger\right)  \rmd W_k(t)
+ \sum_{k=1}^l\left(J_{kt} \sigma_{t_-} J_{kt}^\dagger -\sigma_{t_-} \right) \rmd \widetilde N_k(t).
\end{equation}

Under the physical probability $\Pbb$, the $N_k(t)$  become counting processes of intensities $\lambda_kI_{kt}$ and we set $\rmd \widehat N_k(t)=\rmd N_k(t)-\lambda_kI_{kt} \rmd t$.
We assume to have $I_{kt}\neq 0$. Now, the system of dynamical SDEs \eqref{XnewPa} and
\eqref{nonlinearSME} becomes
\begin{multline}
\rmd X_i(t) =\left(c_i \big(X(t_-)\big)+ \sum_{k=1}^{d}b_{i}^k\big(X(t_-)\big)m_k(t) +  \sum_{k=1}^l g_{i}\big(X(t_-),k\big) \lambda_kI_{kt} \right)\rmd t
\\ \label{XnewPakn}{}
+\sum_{k=1}^{d}b_{i}^k\big(X(t_-)\big) \rmd \widehat W_k(t)+ \sum_{k=1}^l g_{i}\big(X(t_-),k\big) \rmd \widehat N_k(t),
\end{multline}
\begin{equation}\label{knlSME}
\rmd  \widehat\sigma_t = {\Lcal_t}_* \big[\widehat\sigma_{t_-} \big]\rmd t +\sum_{k=1}^d \left( \widehat \sigma_{t_-}{L_{kt}}^\dagger  +L_{kt}\widehat \sigma_{t_-}-m_k(t)\widehat \sigma_{t_-}\right) \rmd \widehat  W_k(t)
+ \sum_{k=1}^l\left(\frac{J_{kt}\widehat \sigma_{t_-} J_{kt}^\dagger}{I_{kt}}-\widehat \sigma_{t_-}\right)\rmd \widehat N_k(t).
\end{equation}
We have written the linear SME \eqref{klSME} and the non-linear one \eqref{knlSME} in strong form because only finite sums are involved and not integrals as in \eqref{linearSME}, \eqref{nonlinearSME}. Recall that the conditional mean of the driving increments, given the past frozen, vanishes; in symbols:
\[
\Ebb_\Pbb\left[\rmd \widehat W_k(t)\Big|\Fscr_{t_-}\right]=0,\qquad \Ebb_\Pbb\left[\rmd \widehat N_k(t)\Big|\Fscr_{t_-}\right]=0.
\]
This is because, under $\Pbb$, $\widehat W_k(t)$ is a Wiener process and $\widehat N_k(t)$ is a compensated counting process; both processes are martingales.

Let us comment the dynamics \eqref{XnewPakn} of the classical component.  By the vanishing of the means of the increments of the driving processes, the terms in the second line of \eqref{XnewPakn} can be interpreted as fluctuating forces; the first term is a continuous contribution and the second one gives also jumps. The quantum component contributes to these two terms only through the probability law of the processes $\widehat N_k(t)$. Instead, in the first line, a drift term appears.  The first term in the drift does not depend on the quantum component and it represents some auto-interaction of the classical component; it could be generated by a classical Hamiltonian function. The two other terms in the drift can again depend on the classical component, as all the coefficients can contain $X(t_-)$, but they depend also on the quantum state through the ``quantum expectations'' $m_{k}(t)$ and $I_{tk}$. They represent a back-reaction of the quantum component on the classical one. Let us stress that the presence of these back-reaction drifts needs the non-vanishing of the fluctuating forces, because they are connected by the coefficients $b^k(x)$ and $g(x,k)$: the classical evolution must be random if there is transfer of information from the quantum component. In the case without jumps, particular cases of \eqref{XnewPa}, including the back-reaction, have been obtained also in other approaches,  see for instance \cite[(49)]{Dio23}, \cite[(7)]{LOW22}. Again similar non-linear back-reaction terms are obtained in completely different schemes in which the hybrid dynamics is obtained from microscopic models via Ehrenfest-like approximations \cite{Pomar+23}.

A way to compare the trajectory approach to hybrid systems to approaches based on moments \cite{BriU23,Pomar+23} would be to compute the multi-time correlations of the classical component by using \eqref{XnewPakn} and stochastic calculus, or by using the characteristic functions of the multi-time probabilities \eqref{genprobs} and the characteristic operator given in \eqref{TtoI}. The computations of the multi-time moments in the case of continuous measurements and purely diffusive case are developed in \cite[Sec.\ 4.3]{BarG09}. For the first moment it is enough to take the mean of \eqref{XnewPakn}; we get directly
\[
\frac{\rmd \ }{\rmd t}\Ebb_\Pbb[ X(t)] =\Ebb_\Pbb\left[c \big(X(t)\big)\right]+ \sum_{k=1}^{d}\Ebb_\Pbb\left[b^k\big(X(t)\big)m_k(t)\right] +  \sum_{k=1}^l \Ebb_\Pbb\left[g\big(X(t),k\big) \lambda_kI_{kt} \right].
\]
Here we recognize immediately the terms discussed above: the mean of the auto-interaction of the classical component and of the two terms giving the back-action of the quantum component on the classical system.

About the dynamics \eqref{knlSME} of the quantum component, we note the dependence on $X(t_-)$ in the Liouville operator \eqref{Lcal*}, \eqref{LJI}. This gives the non-Markovian behaviour of the mean quantum dynamics, already commented after Eq.\ \eqref{redqeq}. When present, this $x$-dependence in \eqref{Lcal*} represents the action of the classical component on the quantum one, say a classical feedback or an external control \cite{WisM10,Till24,BarG12,Jac14}. It can be shown that \eqref{knlSME} preserves the pure states, as it is equivalent to the non-linear SSE recalled in Remark \ref{rem:nlSSE}. Often, this equation is introduced as \emph{stochastic unravelling} of the dissipative dynamics of the mean state \cite{WisM10,Dio23,LOW22}.

Finally, the formal hybrid generator \eqref{Kcal} becomes
\begin{multline}
\Kcal[a\otimes f](x) = f(x)\widetilde\Lcal(x)[a]
+ a \sum_{i=1}^s \frac{\partial f(x)}{\partial x_i}\, c_i (x)+a\sum_{i,j=1}^s\frac 12\, \frac{\partial^2 f(x)}{\partial x_i\partial x_j}\sum_{k=1}^{d}b_{i}^k(x) b_{j}^k(x)
\\  {}+\sum_{i=1}^s \frac{\partial f(x)}{\partial x_i}\sum_{k=1}^{d}b_{i}^k(x)\left(a L_k(x)+L_k(x)^\dagger a \right)
+ \sum_{k=1}^l\lambda_k f\big(x+g(x,k)\big)  J_k(x)^\dagger a J_k(x) ,
\end{multline}
\[
\widetilde\Lcal(x)[a]=
\rmi[H(x),a]+\sum_{k=1}^d\left( L_{k}(x)^\dagger a L_{k}(x)-\frac 12 \left\{L_{k}(x)^\dagger  L_{k}(x),a\right\}\right)
-\frac 12 \sum_{k=1}^l\lambda_k\left\{J_k(x)^\dagger  J_k(x),a\right\}.
\]

As an example, by using the classical component given by \eqref{Diodiscr}, and by taking $L_k(x)=0$, it is possible to construct discrete hybrid master equations, based on the classical Pauli equation, as those introduced in \cite{Dio14,Dio23}.

\subsubsection{Unitary jumps.}\label{sec:Uj}
A particularly simple case is when the jump operators $J_k(x)$ are proportional to unitary operators (as in Sec.\ \ref{sec:qfdyn}) and there is no diffusion contribution; so we take $L_k(x)=0$ and
\begin{gather*}
J_k(x)=\beta_k(x)U_k(x), \quad \beta_k(x)> 0,\quad U_k(x)^\dagger U_k(x)=\openone,
\\ \Rightarrow \quad J_k(x)^\dagger J_k(x)=\beta_k(x)^2\openone, \quad I_k(x)=\beta_k(x)^2.
\end{gather*}
Now, from Eqs.\  \eqref{klSME}, \eqref{Lcal*}, \eqref{knlSME} we get
\begin{equation}
\Lcal(x)_*[\rho]=
-\rmi[H(x),\rho]+\sum_{k=1}^l\lambda_k \beta_k(x)^2\left(U_k(x)\rho U_k(x)^\dagger -\rho\right),
\end{equation}
\begin{equation}
\rmd \sigma_t=  -\rmi[H_t,\sigma_{t_-}]  \rmd t
+ \sum_{k=1}^l\beta_{kt}^{\;2}\left(U_{kt} \sigma_{t_-} U_{kt}^\dagger -\sigma_{t_-} \right) \rmd  N_k(t),\qquad \beta_{kt}=\beta_k\big(X(t_-)\big),
\end{equation}
\begin{equation}\label{Uapos}
\rmd \widehat \sigma_t=-\rmi\left[ H_t,\,\widehat\sigma_{t_-}\right]\rmd t+ \sum_{k=1}^l\left(U_{kt}\widehat \sigma_{t_-} U_{kt}^\dagger-\widehat \sigma_{t_-}\right)\rmd N_k(t),\qquad U_{kt}=U_k\big(X(t_-)\big).
\end{equation}
Note that $I_{kt}=\beta_k\big(X(t_-)\big)^2$ does not depend on the quantum state and this gives that, under $\Pbb$, the full law of the processes $N_k(t)$ remains  independent of the quantum state.

\subsection{Entanglement creation and preservation}
An interesting application of the trajectory formalism is to the entanglement dynamics of two quantum systems without direct interaction; we can have an a priori dynamics \eqref{redqeq} which destroys any initial entanglement, while it is totally or partially protected in a single quantum trajectory \cite{BarGent,VogS10,VGCBB10,MCVF11}. To give an example, we consider the simple dynamics discussed in Sec.\ \ref{sec:Uj}; now, $\Hscr=\Cbb^2\otimes \Cbb^2$. As classical component we take the Poisson processes themselves:
\[
l=s, \qquad X_i(0)=0, \qquad X_i(t)=N_i(t).
\]
With this choice, \eqref{Uapos} gives the dynamics of the conditional state \eqref{apos}; as in the general case, it sends pure states into pure states. So, we take a pure state at time $t=0$ and we can write $\widehat \sigma_t= |\widehat \psi_t\rangle \langle \widehat \psi_t|$.

To quantify the entanglement we consider the concurrence \cite{Woo98}, whose definition is recalled in \ref{app:conc}. Now we have the a priori state and the conditional state, which is random. Due to the definition \eqref{Cmixed} of the concurrence of a mixed state, it is easy to see that the concurrence of the a priori state $\eta_t$ is a lower bound for the mean concurrence of the conditional state $\widehat \psi_t$:
\begin{equation*}
\Ebb_\Pbb \big[ C(\widehat \psi_t)\big]= \Ebb_\Pbb \big[\abs{ \chi(t)}\big] \geq C(\eta_t), \qquad \chi(t)= \chi\big(\widehat \psi_t\big)=\langle \T \widehat\psi_t|\sigma_y^1\sigma_y^2 \widehat\psi_t\rangle .
\end{equation*}
By  \eqref{Uapos}, we get the evolution equation for the random chi-quantity:
\begin{equation}\label{dchi}
\rmd \chi(t)=\left(-\rmi\langle \T  H_t\widehat\psi_t|\sigma_y^1\sigma_y^2 \widehat\psi_t\rangle-\rmi \langle \T \widehat\psi_t|\sigma_y^1\sigma_y^2 H_t\widehat\psi_t\rangle\right)\rmd t
+ \sum_{k=1}^s\left(\langle \T U_{kt}\widehat\psi_t|\sigma_y^1\sigma_y^2 U_{kt}\widehat\psi_t\rangle-\chi(t)\right)\rmd N_k(t).
\end{equation}

\paragraph{Local operators.}
Let us consider the case of only local operators in \eqref{Uapos}, say
\[
H=\frac \omega 2 \left(\sigma^1_z+\sigma^2_z\right),\qquad U_k(x)=U_k= \begin{cases} \rme^{\rmi \phi_k \vec \sigma\cdot \vec \epsilon^k}\otimes \openone \quad  {}& k=1,\ldots,l_1,
\\  \openone\otimes\rme^{\rmi \phi_k \vec \sigma\cdot \vec \epsilon^k} & k=l_1+1,\dots, l, \end{cases}
\]
with $\phi_k\in \Rbb$, \ $\vec\epsilon^k\in \Rbb^3$, \ $\abs{\vec \epsilon^k}=1$. \ By \eqref{dchi}, \eqref{ABX} we get $\rmd \chi(t)=0$: the conditional concurrence is constant in time. If the initial state is maximally entangled, we have $\Ebb_\Pbb \big[ C(\widehat \psi_t)\big]=1$, $\forall t\geq 0$.

On the contrary, the concurrence of the a priori state can decay. For instance let us take
\begin{equation*}
\omega >0, \qquad l=2,\quad l_1=1, \quad \lambda_k=\lambda>0, \quad \beta_k=\beta\neq 0, \quad \phi_k=\frac \pi 2 , \quad \vec \epsilon^k=(1,\, 0,\, 0) \qquad  \Rightarrow \qquad U_k=\rmi \sigma_x^k.
\end{equation*}
Then, by taking the $\Pbb$-mean of Eq.\ \eqref{Uapos}, we get the master equation for the a priori states
\begin{equation}\label{masterxx}
\frac{\rmd \ }{\rmd t}\, \eta_t=\sum_{k=1}^2\left(-\rmi\,\frac \omega 2 \left[\sigma^k_z,\eta_t\right] +\lambda\beta^2\left(\sigma_x^k\eta_t \sigma^k_x-\eta_t\right)\right).
\end{equation}
For this equation there is a unique equilibrium state $\eta_\infty=\openone/4$, which is separable; so, the a priori concurrence decays to zero. In spite of this, if we know the values of $N_1(t)$ and $N_2(t)$ at all times, we know the state of the quantum component, which maintains the initial entanglement.

\paragraph{Non-local operators.}
By taking the same Hamiltonian and the same constants $\lambda $ and $\beta$, but
$U_1=\frac1{\sqrt 2}\left(\sigma_x^1+\rmi\sigma_x^2\right)$ and $U_2=\frac1{\sqrt 2}\left(\sigma_x^1-\rmi\sigma_x^2\right)$, we get again the master equation \eqref{masterxx} for the a priori states. However, in this case, the entanglement of the conditional states can vary; in particular entanglement can be created \cite{BarGent}. For instance, if the initial state is $|11\rangle$, at a jump we have
\[
|11\rangle \to U_k |11\rangle =\frac 1{\sqrt 2}\left(|01\rangle\pm \rmi |10\rangle\right):
\]
we get a maximally entangled state from a separable one.

We have constructed two simple models with the same quantum master equation \eqref{masterxx} for the mean state, a master equation which destroys any initial entanglement. On the contrary, the two decompositions in ensembles of conditional states are different; the first one preserves any initial entanglement, while the second one allows for increasing and decreasing of entanglement. Let us stress that the two decompositions are physically different, as they depend on different interactions with the classical component of the hybrid system.

\paragraph{Hidden entanglement.} For the two models above we can speak of the phenomenon of \emph{hidden entanglement}, a term which is sometimes used for situations in which strong entanglement is present in a statistical ensemble, while it is less relevant in the associated mean state, see for instance \cite{LoFBen14,LoFC18}.

The systems considered in \cite{LoFBen14,LoFC18} are different; typically, they consider a unitary local transformation which hits an entangled system with a probability 1/2, so generating a mixed state. Eventually, revivals of the hidden entanglement can be obtained by local pulses.
We can take this as a suggestion: some kinds of ``one shot'' transformations can be introduced also in our ``Markovian'' framework

For instance, we can have a single shot which hits the second qubit only once at a random time and leaves free the first qubit. Let us take
\[
H=\frac \omega 2 \left(\sigma^1_z+\sigma^2_z\right), \qquad l=1, \qquad X(t)=N(t), \qquad J(x)=g(x)\beta \openone\otimes U, \quad \beta>0,
\]
where $U$ is a unitary operator on $\Cbb^2$ and $g(x)$ is a continuous real function such that $g(0)=1$ and $g(x)=0$ for $x\geq 1$. As $X(t)=N(t)$ takes only integer values, we have that $g\big(N(t)\big)$ behaves as Kroneker delta in $0$, and Eq.\ \eqref{knlSME} gives
\begin{equation*}
\rmd \widehat \sigma_t=-\rmi\left[ H,\,\widehat\sigma_{t_-}\right]\rmd t+ \begin{cases}
\left(U\widehat \sigma_{t_-} U^\dagger-\widehat \sigma_{t_-}\right)\rmd N(t), \quad & \text{if } \ N(t_-)=0,
\\ 0 & \text{if } \ N(t_-)\geq 1;
\end{cases}\end{equation*}
$N(t) $ turns out to be a Poisson process of intensity $\lambda\beta^2$. The function $g(x)$ has been taken to be continuous to have that $J(x)$ satisfies Assumption \ref{Ass:1}. This trick allows to introduce single shots as in \cite{LoFBen14,LoFC18}, with the difference that the time of the jump is the random waiting time of a Poisson process.

\subsection{Control}

Firstly, let us note that we could use a component of the classical process just to introduce some noise in the Hamiltonian $H(x)$, or in some of the operators $L_k(x)$, $J(x,u)$. If this component is deterministic, simply we get a time dependent Hamiltonian. By this, we have that at least open loop control can be included in the Markovian hybrid dynamics.

Moreover, we have that the classical subsystem satisfies Eq.\ \eqref{XnewPa} under the physical probability $\Pbb$; so, its dynamics depends on the state of the quantum subsystem through the coefficients $C^{(i)}(t)$. On the other side, the classical component can affect the quantum dynamics through the $x$-dependence of the operators $H(x)$, $L_k(x)$, $J(x,u)$. In this way we have realized a closed loop feedback; applications in quantum optics can be found in \cite{BarG09,WisM10,Jac14}. This opens the question of the connections with a full quantum control theory, where already the quantum trajectory formalism (quantum filtering) started to be explored, see \cite{Ticozzi17,Ticozzi23,Till24} and references therein.

\section{Conclusions}\label{sec:concl}

In the present approach, the hybrid dynamics has been introduced by a couple of stochastic differential equations, one for the classical component and one for the quantum component, and this has been done in two different ways. In the first case a reference probability $\Qbb$ is used, and the two stochastic equations are the evolution equation \eqref{Xprocess} for the classical process and the linear SME \eqref{linearSME} for the (non-normalized) quantum state; the quantum/classical interaction is contained only in the SME \eqref{linearSME}. The physical probability $\Pbb$ is determined by the dynamics itself, and it is defined  in    \eqref{def:pt}, \eqref{physprob}. By suitable hypothesis (Assumptions \ref{Ass:processes}, \ref{ass:classical}, \ref{Ass:1}) it is possible to prove existence and uniqueness of the solution of these coupled equations.

An equivalent description can be given by working under the physical probability $\Pbb$; now we have the stochastic equation \eqref{XnewPa} for the classical process and the non-linear SME \eqref{nonlinearSME} for the quantum states. However, the mathematical setting is more involved as even the probability law depends on  the solution of the two coupled equations. On the other side, in this formulation the quantum/classical interaction is more transparent, as it appears in both equations.

We have also shown that the general axiomatic and the linear structure of a quantum theory is respected. Indeed, we have shown how to connect the whole construction to the general notion of quantum observable (the positive operator valued measures) and to the notion of ``instrument'' (Sec.\ \ref{sec:transinstr}). We have shown that, as expected, the extraction of information from the quantum subsystem is necessarily connected to the presence of suitable dissipation in the dynamics (see the comments in Sec.\ \ref{sec:nprob}, after Eq.\ \eqref{XnewP}). The absence of memory is underlined by the fact that probabilities and dynamics can be summarized in a notion of ``transition instruments'' \eqref{I(E|x)}, a quantum analogue of the transition probabilities in the theory of classical Markov processes.

Moreover, the Markovian character of the dynamics appears also in the fact that it can be expressed as a hybrid dynamical semigroup. However, in the pure classical case it is shown that the classical states can not always be reduced to densities with respect to the Lebesgue measure (\ref{app:classM}); then, the hybrid semigroup needs a $C^*$-algebraic formulation. In any case, we have given its formal generator, which allows to compare our results with less general approaches based on master equations.

Finally, some simple examples on entanglement and control show that we can have a considerable freedom in the construction of physical models,  even under the restriction of a Markovian hybrid dynamics.

\appendix

\section{The classical Markov process under the reference probability $\Qbb$}\label{app:classM}
In this appendix we collect known results on the Markov process defined by \eqref{Xprocess}, when there is no interaction with the quantum component; in particular, we discuss the associated semigroup.
Equation \eqref{Xprocess} is studied in \cite[Sec.\ 6.2]{App09}.

By following \cite[Secs.\ 6.4.1, 6.4.2]{App09}, we can introduce the \emph{stochastic flow} $\Phi$ associated to the stochastic differential equation \eqref{Xprocess}.
For $0\leq s\leq t$, let $\Phi^s_t(x)$ be the solution of \eqref{Xprocess} with initial condition $x$ at time $s$; in particular, with the notation introduced at the end of Proposition \ref{prop:clpr}, we have $\Phi^0_t(x)=X^x(t)$. The solution of \eqref{Xprocess} is an homogeneous Markov process and we have
\begin{equation}\label{transf}
P(t,x,\Eo)\equiv \Qbb[\Phi^0_t(x)\in\Eo] =\Qbb[\Phi^s_{s+t}(x)\in\Eo].
\end{equation}
Then, the transition probabilities \eqref{transprob} satisfy the properties of a \emph{time-homogeneous transition function} \cite[Chapt.\ 4 Sec.\ 1 p.\ 156]{EthK05}: \ (a)
$P(t,x,\cdot)$ is a probability measure on $\Bcal(\Rbb^s)$, $\forall t \in [0,+\infty)$, $\forall x\in \Rbb^s$; \ (b) $P(0,x,\cdot)=\delta_x(\cdot)$, $\forall x\in \Rbb^s$; \ (c) $P(\cdot,\cdot,\Eo)$ is a bounded Borel function on $[0,+\infty)\times \Rbb^s $, $\forall \Eo\in \Bcal(\Rbb^s)$; \ (d) the \emph{Chapman-Kolmogorov identity} holds, which is
\begin{equation}\label{CKeq}
P(t+r,x,\Eo)=\int_{\Rbb^s}P(r,y,\Eo)P(t,x,\rmd y) \qquad r,t\geq 0, \quad x\in\Rbb^s, \quad \Eo\in \Bcal(\Rbb^s).
\end{equation}
We also define the map
\begin{equation}\label{def:Ttcl}
T_t[f](x)=\Ebb_\Qbb\left[f\big(\Phi^0_t(x)\big)\right]=\int_{\Rbb^s}f(y)P(t,x,\rmd y), \qquad f\in B_{\rm b}(\Rbb^s), \quad x\in\Rbb^s;
\end{equation}
$\{T_t,\,t\geq 0\}$ is a semigroup of positive, bounded linear maps from $B_{\rm b}(\Rbb^s)$ into itself, with $T_t(1)=1$, $T_0=\id$. Moreover, $T_t$ maps $C_{\rm b}(\Rbb^s)$ into itself \cite[Note 3 p. 402]{App09}. The expression \eqref{def:Ttcl} is continuous in $t$, for fixed $f$ and $x$.

Let $f$ be in the space $C_{\rm b}^2(\Rbb^s)$ of the two times differentiable functions; \marginlabel{$C_{\rm b}^2(\Rbb^s)$} the functions and their first two  derivatives are continuous and bounded. By \eqref{Xprocess} and It\^o's formula, we get
\begin{multline}\label{dfX}
\rmd f\big(X(t)\big)\Big|_{X(t_-)=x}=\sum_{i=1}^s \frac{\partial f(x)}{\partial x_i}\left( c_i (x)\rmd t +\sum_{k=1}^{d}b_{i}^k(x) \rmd W_k(t)\right)
+\frac 12 \sum_{i,j=1}^s \frac{\partial^2 f(x)}{\partial x_i\partial x_j}\sum_{k=1}^{d}b_{i}^k(x) b_{j}^k(x)\rmd t
\\ {} +\rmd t \int_{\mathsf{U}_0}\left(f\big(x+g(x,u)\big) -f(x)-\sum_{i=1}^s \frac{\partial f(x)}{\partial x_i}\, g_i(x,u)\right)\nu(\rmd u)
\\ {} +\int_{\mathsf{U}_0}\left(f\big(x+g(x,u)\big) -f(x)\right)\widetilde \Pi(\rmd u ,\rmd t) +\int_{\Uscr\setminus \mathsf{U}_0}\left(f\big(x+g(x,u)\big) -f(x)\right) \Pi(\rmd u ,\rmd t).
\end{multline}
Then, we have
\begin{equation}\label{xgenerator}
\frac{\rmd \ }{\rmd t}\,T_t[f](x)= \big(T_t\circ\Kcal^{\rm cl}\big)[f](x),\qquad \forall f\in C_{\rm b}^2(\Rbb^s), \quad \forall x \in \Rbb^s,
\end{equation}
\begin{multline}\label{Kcl}
\Kcal^{\rm cl}[f](x)=\sum_{i=1}^s \frac{\partial f(x)}{\partial x_i}\,c_i (x)+\frac 12 \sum_{i,j=1}^s \frac{\partial^2 f(x)}{\partial x_i\partial x_j}\sum_{k=1}^{d}b_{i}^k(x) b_{j}^k(x)
\\ {} +\int_{\Uscr}\left(f\big(x+g(x,u)\big) -f(x)-\ind_{\mathsf{U}_0}(u)\sum_{i=1}^s \frac{\partial f(x)}{\partial x_i}\, g_i(x,u)\right)\nu(\rmd u).
\end{multline}
Under some additional sufficient assumptions, it is possible to prove that $T_t$ is strongly continuous and that it is generated by $\Kcal^{\rm cl}$, cf.\ \cite[(6.36) p.\ 402]{App09}. We do not need these results, because existence and uniqueness of the solution of \eqref{Xprocess} is enough to give the intrinsic dynamics of the classical component of the hybrid system. In \cite[pp.\ 56, 349-350]{Prott04}, \eqref{xgenerator} is used as definition of ``generator of the time-homogeneous Markov process''; the law of the process is determined by the \emph{Dynkin's expectation formula}
\begin{equation}
\Ebb_\Qbb \left[f\big(X(t)\big)- f\big(X(0)\big)- \int_0^t \Kcal^{\rm cl}[f]\big(X(r)\big)\rmd r\right]=0.
\end{equation}

In some cases the adjoint of $T_t$, defined in the dual Banach space, leaves invariant the space of the probability densities with respect to the Lebesgue measure and some known classical master equations are obtained. In the deterministic case, i.e. $b_{i}^k(x)=0$ and $g(x,u)=0$, we can take $s$ even and the coefficients $c_i(x)$ giving an Hamiltonian evolution; then, the master equation is the Liouville equation. Instead, with only $g(x,u)=0$, we can get the Kolmogorov-Fokker-Planck equation \cite[Sec.\ 3.5.3]{App09}, \cite[Secs.\ 9.3, 14.2.2]{DaPZ14}.

In general, it can happen that $X(t)$ has not a density with respect to Lebesgue measure, even if we start with $X(0)$ having a density. As an example, take the equation $\rmd X(t)= \left[x^0 - X(t_-)\right]\rmd N(t)$ with $x^0\in \Rbb^s$; $X(0)$ is taken with a density and independent of the Poisson process $N$ of intensity $\lambda>0$. Then, we obtain $P[X(t)=x^0]= 1-\rme^{-\lambda t}$ and a density with respect to the Lebesgue measure does not exist for $t>0$.

By generalizing the example above, we can construct models with jumps among fixed points at random times. Let us take $n$ distinct points $x^j\in \Rbb^s$ and $l$ independent Poisson processes $N_k$ of intensities $\lambda_k>0$,
\begin{equation}\label{Diodiscr}\begin{split}
  \rmd X(t)= &\sum_{k=1}^l g\big(X(t_-), k\big)\rmd N_k(t), \qquad g(x,k)=\begin{cases} y^k-x & \text{if } \ x\in \Fo_k, \\ 0 & \text{if } \ x\notin \Fo_k ,\end{cases}
\\
&\Fo_k\in \Bcal(\Rbb^s), \quad \bigcup_{k=1}^l \Fo_k=\Rbb^s, \quad y^k\in \{x^1,\ldots,x^n\}.
\end{split}\end{equation}
After the first jump, the process reduces to a random walk on the points $\{x^1,\ldots,x^n\}$. When the initial condition has a probability concentrated in the set of points $\{x^1,\ldots,x^n\}$, the probability at time $t$ satisfies a \emph{Pauli rate equation} \cite[(13)]{Dio14}, \cite[(14)]{Dio23}.

\section{Concurrence}\label{app:conc}
To quantify the entanglement of two qubits we use the \emph{concurrence} \cite{Woo98}. Let $\varphi\in\Hscr\equiv \Cbb^2\otimes \Cbb^2$ be a generic normalized vector and expand it on the computational basis:
$
\varphi = \varphi_{11}|11\rangle + \varphi_{10}|10\rangle + \varphi_{01}|01\rangle +\varphi_{00}|00\rangle$.
Let $\T$ be the complex conjugation of the coefficients in this basis:
$
\T \varphi = \overline{\varphi_{11}} \,|11\rangle + \overline{\varphi_{10}} \,|10\rangle +
\overline{\varphi_{01}} \,|01\rangle +\overline{\varphi_{00}} \,|00\rangle$.
Then, the \emph{concurrence} $C(\varphi)$  of the pure state
$\varphi$ is defined by
\begin{equation}\label{Cphi}
C(\varphi) : = \abs{  \chi(\varphi)}, \qquad \chi(\varphi) : = \langle \T \varphi | \sigma _y \otimes \sigma _y \varphi \rangle
=2 \left(\varphi_{10}\varphi_{01} - \varphi_{11}\varphi_{00}\right).
\end{equation}
The concurrence goes from 0 (for a separable state) to 1 (for a Bell state).

If $\rho$ is a generic statistical operator, the concurrence is defined by
\begin{equation}\label{Cmixed}
C(\rho): = \inf \sum_{i}p_i C(\psi _{i}),
\end{equation}
where the infimum is taken over all decompositions of $\rho$ into pure states,  $\rho
= \sum_{i}p_i | \psi _i\rangle \langle \psi _i|$.

Finally, let $A$ and $B$ be linear operators on $\Cbb^2$. In studying the dynamics of the concurrence, the following formulae are very useful:
\begin{equation}\label{ABX}\begin{split}
\chi\big((A\otimes B)\varphi\big)&=
\left({\det}_{\Cbb^2} A\right) \left({\det}_{\Cbb^2} B\right) \chi(\varphi),
\\
\langle \T \varphi|(\sigma_y A)\otimes \sigma_y \varphi \rangle &=
\langle \T A\otimes \openone\varphi|\sigma_y \otimes \sigma_y \varphi \rangle =
\frac 1 2 \left( {\Tr}_{\Cbb^2} A\right) \chi(\varphi).\end{split}
\end{equation}


\begin{thebibliography}{99}
\bibitem{DGS00} L. Di\'osi, N. Gisin,  W.T. Strunz, \textsl{Quantum approach to coupling classical and quantum dynamics}, \href{https://journals.aps.org/pra/abstract/10.1103/PhysRevA.61.022108}{Phys. Rev. A \textbf{61} (2000) 022108.}

\bibitem{Dio14} L. Di\'osi, \textsl{Hybrid quantum-classical master equations}, \href{https://doi.org/10.1088/0031-8949/2014/T163/014004}{Phys. Scr. \textbf{T163} (2014) 014004.} 

\bibitem{Dio23} L. Di\'osi, \textsl{Hybrid completely positive Markovian quantum-classical dynamics},  \href{https://journals.aps.org/pra/abstract/10.1103/PhysRevA.107.062206}{Phys. Rev. A \textbf{107} (2023)  062206}; \textsl{Erratum}, \href{https://doi.org/10.1103/PhysRevA.108.059902}{Phys. Rev. A \textbf{108} (2023) 059902(E)}.

\bibitem{Op+22} J. Oppenheim, C. Sparaciari, B. Šoda, Z. Weller-Davies, \textsl{The two classes of hybrid classical-quantum dynamics}, \href{https://doi.org/10.48550/arXiv.2203.01332}{arXiv:2203.01332 [quant-ph] (2022)}.

\bibitem{LOW22} I. Layton, J. Oppenheim, Z. Weller-Davies, \textsl{A healthier semi-classical dynamics}, \href{https://doi.org/10.48550/arXiv.2208.11722}{arXiv:2208.11722 [quant-ph]} (2022).

\bibitem{Opp+23} J. Oppenheim, C. Sparaciari, B. Šoda, Z. Weller-Davies, \textsl{Objective trajectories in hybrid classical-quantum dynamics},
\href{https://doi.org/10.22331/q-2023-01-03-891}{Quantum \textbf{7} (2023)  891}. 

\bibitem{DamWer23} L. Dammeier, R.F. Werner, \textsl{Quantum-classical hybrid systems and their quasifree transformations}, \href{https://doi.org/10.22331/q-2023-07-26-1068}{Quantum \textbf{7} (2023) 1068}.

\bibitem{Sergi+23} A. Sergi, D. Lamberto, A. Migliore, A. Messina, \textsl{Quantum–classical hybrid systems and Ehrenfest's Theorem}, \href{https://doi.org/10.3390/e25040602}{Entropy \textbf{25} (2023) 602}. 

\bibitem{ManRT23} G. Manfredi, A. Rittaud, C. Tronci, \textsl{Hybrid quantum-classical dynamics of pure-dephasing systems}, \href{https://doi.org/10.1088/1751-8121/acc21e}{J. Phys. A: Math. Theor. \textbf{56} (2023) 154002}. 

\bibitem{Tronci23} W. Bauer, P. Bergold, F. Gay-Balmaz, C. Tronci, \textsl{Koopmon trajectories in nonadiabatic quantum-classical dynamics}, \href{https://doi.org/10.48550/arXiv.2312.13878}{arXiv:2312.13878[quant-ph] (2023)}.

\bibitem{Pomar+23} J.L. Alonso, C. Bouthelier-Madre, J. Clemente-Gallardo, D. Mart\'inez-Crespo, J. Pomar, \textsl{Effective nonlinear Ehrenfest hybrid quantum-classical dynamics}, \href{https://doi.org/10.1140/epjp/s13360-023-04266-w}{Eur. Phys. J. Plus \textbf{138} (2023) 649}.

\bibitem{BriU23} D. Brizuela,  S.F. Uria, \textsl{Hybrid classical-quantum systems in terms of moments}, \href{https://doi.org/10.1103/PhysRevA.109.032209}{Phys. Rev. A \textbf{109} (2024) 032209}.

\bibitem{BarW23}  A. Barchielli, R. Werner, \textsl{Hybrid quantum-classical systems: Quasi-free Markovian dynamics}, \href{https://doi.org/10.1142/S0219749924400021}{Int. J. Quantum Inf. (2024) DOI: 10.1142/S0219749924400021}.

\bibitem{Bar23} A. Barchielli, \textsl{Markovian master equations for quantum-classical hybrid systems}, \href{https://doi.org/10.1016/j.physleta.2023.129230}{Phys. Lett. A \textbf{492} (2023) 129230}.

\bibitem{Bar86} A.~Barchielli, \textsl{Measurement
    theory and stochastic differential equations in quantum mechanics}, \href{http://pra.aps.org/abstract/PRA/v34/i3/p1642_1}{Phys.\ Rev.\ A {\bf 34} (1986) 1642--1649}.

\bibitem{Bel88} V.P. Belavkin, \textsl{Nondemolition measurements, nonlinear filtering and dynamic programming of quantum stochastic processes}. In A. Blaqui\`ere (ed.), \textit{Modelling and Control of Systems}, \href{https://link.springer.com/chapter/10.1007/BFb0041197#citeas}{Lecture Notes in Control and Information Sciences, vol.\ 121 (Springer, Berlin, 1988) pp.\ 245--265}.

\bibitem{Bel89} V.P. Belavkin, \textsl{A new wave equation for a continuous nondemolition measurement}, \href{https://doi.org/10.1016/0375-9601(89)90066-2}{Phys. Lett. A \textbf{140} (1989) 355-358}.

\bibitem{BarB91} A. Barchielli, V.P. Belavkin,    \textsl{Measurements    continuous in time and   a posteriori states in quantum mechanics},    \href{http://iopscience.iop.org/0305-4470/24/7/022}{J.\ Phys.\ A: Math.\ Gen. {\bf24} (1991) 1495--1514.}

\bibitem{Bar93} A. Barchielli,    \textsl{On the    quantum theory of measurements continuous in time}, \href{http://dx.doi.org/10.1016/0034-4877(93)90037-F}{Rep.\ Math.\    Phys.\ {\bf 33} (1993) 21--34.}

\bibitem{ZolG97}  P. Zoller and C.W. Gardiner, \textsl{Quantum noise in quantum optics: the stochastic Schr\"odinger equation}. In S. Reynaud, E. Giacobino \& J. Zinn-Justin eds., \textit{Fluctuations quantiques, (Les Houches 1995)} (North-Holland, Amsterdam, 1997) pp. 79--136.

\bibitem{BarPZ98} A. Barchielli, A.M. Paganoni, F. Zucca, \textsl{On stochastic differential equations and semigroups of probability operators in quantum probability}, \href{http://dx.doi.org/10.1016/S0304-4149(97)00093-8}{Stoch. Proc.\ Appl.\ {\bf73} (1998)}.
    69--86.

\bibitem{Hol01} A.S. Holevo, \textit{Statistical Structure of Quantum Theory}, Lecture Notes in Physics m 67 (Springer, Berlin, 2001).

\bibitem{BarG09} A. Barchielli, M. Gregoratti,     \textit{Quantum Trajectories and Measurements in Continuous Time: The Diffusive Case}, \href{http://www.springer.com/physics/quantum+physics/book/978-3-642-01297-6}{Lect.\ Notes Phys.\ \textbf{782} (Springer, Berlin \& Heidelberg,  2009)}.

\bibitem{WisM10} H.M. Wiseman and G.J. Milburn, \textit{Quantum Measurement and  Control} (Cambridge University Press, Cambridge, 2010).

\bibitem{Bar06} A. Barchielli, \textsl{Continual Measurements in Quantum Mechanics and  Quantum Stochastic Calculus}.    In S. Attal, A. Joye, C.-A. Pillet (eds.),   \textit{Open Quantum Systems III}, \href{http://www.springerlink.com/content/2623j57n1t7w0131/}{Lect.\ Notes Math.\ \textbf{1882}    (Springer, Berlin, 2006), pp. 207--291}.


\bibitem{BGM04} L. Bouten, M. Guţă, H. Maassen, \textsl{Stochastic Schr\"odinger equations}, \href{https://doi.org/10.1088/0305-4470/37/9/010}{ J. Phys. A: Math. Gen. \textbf{37} (2004) 3189.}

\bibitem{Maa23} H. Maassen,  \textsl{Continuous observation of quantum systems}, \href{https://doi.org/10.1142/S0219749924400112}{Int. J. Quantum Inf. (2024) DOI: 10.1142/S0219749924400112}

\bibitem{Till24} A. Tilloy, \textsl{General quantum-classical dynamics as measurement based feedback}, \href{https://doi.org/10.48550/arXiv.2403.19748}{arXiv:2403.19748 [quant-ph] (2024).}

\bibitem{Mim24} M.A. Mimona, M.H. Mobarak, E. Ahmed, F. Kamal,  M. Hasan, \textsl{Nanowires: Exponential speedup in quantum computing}, \href{https://doi.org/10.1016/j.heliyon.2024.e31940}{Heliyon \textbf{10} (2024) e31940}.

\bibitem{OppWD22} J. Oppenheim, Z. Weller-Davies, \textsl{The constraints of post-quantum classical gravity},  \href{https://doi.org/10.1007/JHEP02(2022)080 (2022)}{J. High Energ. Phys. 2022, 80 (2022)}. 

\bibitem{Opp+23grav} I. Layton, J. Oppenheim, A. Russo, Z. Weller-Davies, \textsl{The weak field limit of quantum matter back-reacting on classical spacetime}, \href{https://doi.org/10.1007/JHEP08(2023)163}{J. High Energ. Phys. 2023, 163 (2023)}. 

\bibitem{ProsB23}    G.M. Prosperi, M. Baldicchi, \textit{Interpretation of Quantum Theory and Cosmology}, \href{https://doi.org/10.48550/arXiv.2304.07095}{arXiv:2304.07095 (2023)}

\bibitem{HallR18} M.J.W. Hall, M. Reginatto, \textsl{On two recent proposals for witnessing nonclassical gravity}, \href{https://iopscience.iop.org/article/10.1088/1751-8121/aaa734}{J.  Phys. A: Math. Theor.
    \textbf{51} (2018) 085303}. 

\bibitem{LoFC18} R. Lo Franco, G. Compagno, \textsl{Overview on the phenomenon of two-qubit entanglement revivals in classical environments}, in F. Fernandes Fanchini, D. de Oliveira Soares Pinto, G. Adesso, \textit{Lectures on General Quantum Correlations and their Applications} \href{https://link.springer.com/chapter/10.1007/978-3-319-53412-1_17}{(Springer, Cham, 2017) pp.\ 367--391}.

\bibitem{BarGent} A. Barchielli, M. Gregoratti, \textsl{Entanglement protection and    generation under continuous monitoring}, in L.\ Accardi, F.\ Fagnola, \textit{Quantum    Probability and Related Topics}, QP-PQ: Quantum Probability and White    Noise Analysis, \href{http://dx.doi.org/10.1142/9789814447546_0002}{Vol.\ 29, (World Scientific, Singapore, 2013) pp.\     17--42}.

\bibitem{BarH95} A. Barchielli, A.S. Holevo, \textsl{Constructing
    quantum measurement processes via classical stochastic calculus}, \href{http://dx.doi.org/10.1016/0304-4149(95)00011-U}{Stoch. Proc. Appl. {\bf 58} (1995) 293--317}.

\bibitem{BarG12}  A. Barchielli, M. Gregoratti, \textsl{Quantum measurements in continuous time, non-Markovian evolutions and feedback}, \href{http://rsta.royalsocietypublishing.org/content/370/1979/5364.abstract}{Phil. Trans.   R. Soc. A \textbf{370} no. 1979 (2012) 5364--5385}.

\bibitem{Met82} M. M\'etivier, \textit{Semimartingales, a Course on Stochastic Processes} (W. de Gruyter, Berlin, 1982).

\bibitem{IWat89} N. Ikeda, S. Watanabe, \textit{Stochastic Differential Equations and Diffusion Processes}, Second Edition (North-Holland Publishing Company, Amsterdam, 1989).

\bibitem{LipSmart}
R.Sh. Liptser, A.N. Shiryayev, \textit{Theory of Martingales}, (Kluwer Academic Publishers, Dordrecht, 1986).

\bibitem{DaPZ14} G. Da Prato, J. Zabczyk, \textit{Stochastic Equations in Infinite Dimensions}, II Edition (Cambridge University Press, Cambridge, 2014).

\bibitem{Prott04} P.E. Protter, \textit{Stochastic Integration and Differential Equations}, II Edition (Springer, Berlin, 2004)

\bibitem{EthK05} S.N. Ethier, T.G. Kurtz, \textit{Markov Processes --- Characterization and Convergence} (Wiley, Hoboken, New Jersey, 2005).

\bibitem{App09} D. Applebaum, \textit{L\'evy Processes and Stochastic Calculus},
Second Edition (Cambridge University Press, 2009).

\bibitem{Sato99} K. Sato, \textit{L\'evy processes and infinitely divisible distributions} (Cambridge University Press, Cambridge, 1999)

\bibitem{Hol88QPIII} A.S. Holevo, \textsl{A noncommutative generalization of conditionally positive definite functions}, in L.~Accardi, W.~von Waldenfels (eds.), \textit{Quantum Probability and Applications III}. \href{https://link.springer.com/chapter/10.1007/BFb0078059}{Lecture Notes in  Mathematics, Vol.\ 1303 (Springer, Berlin, 1988), pp.\ 128--148}. 

\bibitem{LoFBen14} A. D'Arrigo, R. Lo Franco, G. Benenti, E. Paladino, G. Falci, \textsl{Recovering entanglement by local operations}, \href{https://doi.org/10.1016/j.aop.2014.07.021}{Ann. Phys. 350 (2014) 211--224}. 

\bibitem{BarL08} A. Barchielli, G. Lupieri, \textsl{Information  gain in quantum continual measurements}, in V.P. Belavkin and M.
    Gu\c{t}\v{a},  \textit{Quantum    Stochastic and Information} \href{http://dx.doi.org/10.1142/9789812832962_0015}{(World Scientific, Singapore, 2008) pp.\    325--345}.

\bibitem{Xi08}J. Xiong, \textit{An Introduction to Stochastic Filtering Theory}, vol. 18, OUP Oxford, 2008.

\bibitem{Hol12} A.S. Holevo, \textit{Quantum Systems, Channels, Information} (De Gruyter, Berlin, 2012).

\bibitem{Hol86} A.S. Holevo,  \textsl{Conditionally positive definite functions in quantum probability}, Proc. Int. Cong. Math. 1011-1020 (1986) Berkeley.

\bibitem{Hol89} A.S. Holevo, \textsl{Limit theorems for repeated measurements and continuous measurement processes}, in L. Accardi, W. von Waldenfels (eds.), \textit{Quantum    Probability and Applications IV}, \href{https://link.springer.com/chapter/10.1007/BFb0083555}{Lecture Notes in Mathematics \textbf{1396}  (Springer,
    Berlin, 1989) pp. 229--257}. 

\bibitem{Bar89} A. Barchielli, \textsl{Probability operators and convolution semigroups of  instruments in quantum
    probability}, \href{https://link.springer.com/article/10.1007/BF00340008}{Probab.\ Theory Rel.\ Fields {\bf 82} (1989) 1--8}.

\bibitem{BarL91}  A. Barchielli, G. Lupieri, \textsl{A quantum analogue of Hunt's representation theorem for the generator of convolution semigroups on Lie groups}, \href{https://doi.org/10.1007/BF01212558}{Probab.\ Theory Rel.\ Fields \textbf{88} (1991) 167--194}.

\bibitem{BarHL93} A. Barchielli, A.S. Holevo, G. Lupieri,  \textsl{An  analogue of Hunt's representation theorem in quantum probability}, \href{https://link.springer.com/article/10.1007/BF01047573} {J.\  Theor.\ Probab.\ {\bf 6} (1993) 231--265}. 

\bibitem{BarP96} A. Barchielli, A.M. Paganoni,   \textsl{A  note on a formula of L\'evy--Khinchin type in quantum probability}, \href{https://doi.org/10.1017/S0027763000005511} {Nagoya Math.\ J.\ {\bf141} (1996) 29--43}. 

\bibitem{DVV79} B. Demoen, P. Vanheuverzwijn, A. Verbeure, \textsl{Completely positive quasi-free maps of the CCR-algebra}, \href{https://doi.org/10.1016/0034-4877(79)90049-1}{Rep.\ Math.\ Phys.\ \textbf{15} (1979) 27--39}. 

\bibitem{Jac14} K. Jacobs, \textit{Quantum Measurement Theory and its Applications} \href{https://doi.org/10.1017/CBO9781139179027}{(Cambridge University Press, 2014)}.

\bibitem{VogS10} S. Vogelsberger, D. Spehner, \textsl{Average entanglement for  Markovian quantum trajectories}, \href{https://doi.org/10.1103/PhysRevA.82.052327}{Phys. Rev. A \textbf{82} (2010) 052327}. 

\bibitem{VGCBB10} C. Viviescas, I. Guevara, A.R.R. Carvalho, M. Busse, A. Buchleitner, \textsl{Entanglement dynamics in open two-qubit systems via
    diffusive quantum trajectories}, \href{https://doi.org/10.1103/PhysRevLett.105.210502}{Phys. Rev. Lett. \textbf{105} (2010) 210502}.

\bibitem{MCVF11} E.  Mascarenhas, D. Cavalcanti, V. Vedral, M. Fran\c ca Santos, \textsl{Physically realizable entanglement by local continuous
    measurements}, \href{https://doi.org/10.1103/PhysRevA.83.022311}{Phys. Rev. A \textbf{83} (2011) 022311}. 

\bibitem{Woo98} W.K. Wootters, \textsl{Entanglement of formation of an arbitrary state of two qubits}, \href{https://doi.org/10.1103/PhysRevLett.80.2245}{Phys. Rev. Lett. \textbf{80} (1998) 2245-2248}.


\bibitem{Ticozzi17} T. Benoist, C. Pellegrini, F. Ticozzi, \textsl{Exponential stability of subspaces for
 quantum stochastic master equations}, \href{http://dx.doi.org/10.1007/s00023-017-0556-3}{Ann. Henri Poincar\'e 18 (2017) 2045–2074}.

\bibitem{Ticozzi23} W. Liang, K. Ohki, F. Ticozzi, \textsl{On the Robustness of Stability for Quantum Stochastic Systems}, \href{http://dx.doi.org/10.1109/CDC49753.2023.10383791}{62nd IEEE Conference on Decision and Control (CDC) (2023), pp.\ 7202-7207}.
\end{thebibliography}
\end{document}